%% file: main.tex
\documentclass{article}

\usepackage{microtype}
\usepackage{graphicx}
\usepackage{booktabs} % for professional tables

\usepackage{hyperref}

\usepackage[accepted]{icml2024}

\usepackage{amsmath}
\usepackage{amssymb}
\usepackage{mathtools}
\usepackage{amsthm}

\usepackage[capitalize,noabbrev]{cleveref}

\theoremstyle{plain}
\newtheorem{theorem}{Theorem}[section]
\newtheorem{proposition}[theorem]{Proposition}
\newtheorem{lemma}[theorem]{Lemma}
\newtheorem{corollary}[theorem]{Corollary}
\theoremstyle{definition}

\theoremstyle{remark}

\usepackage[textsize=tiny]{todonotes}

\icmltitlerunning{Learning High-Order Relationships of Brain Regions}

\usetikzlibrary{positioning}
\newcommand\shline{\addlinespace[3pt]\hline\addlinespace[3pt]}
\usepackage{tabularx}
\usepackage{caption}
\usepackage{subcaption}
\usepackage{multirow}
\usepackage{pifont}
\usepackage{enumitem}
\usepackage{wrapfig}
\usepackage{bbm}

\newcommand{\mname}{\textsc{HyBRiD}}
\newcommand{\multiheadname}{\textit{multi-head drop-bottleneck}}

\newcommand{\avgimp}{$12.1\%$}

\begin{document}

\twocolumn[
\icmltitle{Learning High-Order Relationships of Brain Regions}



\icmlsetsymbol{equal}{*}

\begin{icmlauthorlist}
\icmlauthor{Weikang Qiu}{yale}
\icmlauthor{Huangrui Chu}{yale}
\icmlauthor{Selena Wang}{yale}
\icmlauthor{Haolan Zuo}{yale}
\icmlauthor{Xiaoxiao Li}{ubc,vector}
\icmlauthor{Yize Zhao}{yale}
\icmlauthor{Rex Ying}{yale}
\end{icmlauthorlist}

\icmlaffiliation{ubc}{University of British Columbia, Vancouver, Canada}
\icmlaffiliation{vector}{Vector Institute, Toronto, Canada}
\icmlaffiliation{yale}{Yale University, New Haven, USA}

\icmlcorrespondingauthor{Rex Ying}{rex.ying@yale.edu}
% \icmlcorrespondingauthor{}{first1.last1@ece.ubc.ca}
% \icmlcorrespondingauthor{Firstname2 Lastname2}{first2.last2@yale.edu}

\icmlkeywords{Machine Learning, ICML}

\vskip 0.3in
]



\printAffiliationsAndNotice{}  % leave blank if no need to mention equal contribution

\begin{abstract}
Discovering reliable and informative relationships among brain regions from functional magnetic resonance imaging (fMRI) signals is essential in phenotypic predictions. Most of the current methods fail to accurately characterize those interactions because they only focus on pairwise connections and overlook the high-order relationships of brain regions. We propose that these high-order relationships should be \textit{maximally informative and minimally redundant} (MIMR).
However, identifying such high-order relationships is challenging and under-explored due to the exponential search space and the absence of a tractable objective. %Methods that can be tailored to our context are also non-existent.
In response to this gap, we propose a novel method named \mname{} which aims to extract MIMR high-order relationships from fMRI data. \mname{} employs a \textsc{Constructor} to identify hyperedge structures, and a \textsc{Weighter} to compute a weight for each hyperedge, which avoids searching in exponential space. \mname{} achieves the MIMR objective through an innovative information bottleneck framework named \multiheadname{} with theoretical guarantees. Our comprehensive experiments demonstrate the effectiveness of our model. Our model outperforms the state-of-the-art predictive model by an average of $11.2\%$, regarding the quality of hyperedges measured by CPM, a standard protocol for studying brain connections. Source code is available at \url{https://github.com/Graph-and-Geometric-Learning/HyBRiD}.

% \jialin{However statement is not strong enough. I think exponential search space is not the reason for under-exploring. We can say: that identifying them is difficult since (1) ideal high-order relationships should maximize ... and minimize ... and (2) exponential search space... Then we emphasize that there exists no method ...}

% Beyond these technical achievements, our method sheds light on understanding human cognition and behavioral traits.
% Furthermore, although there is a notable correlation between this challenge and the information bottleneck (IB) paradigm, methods tailored to our specific context are non-existent.
% Although there are some existing methods to identify high-order relationships, they struggle to pinpoint the ones that are truly informative and fail to learn consistent structures across subjects. Furthermore, although there is a notable correlation between this challenge and the information bottleneck (IB) paradigm, methods tailored to our specific context are non-existent.

\end{abstract}

\input{000intro}
\input{002preliminaries}
\input{001related_work}
\input{003approach}
\input{004experiments}
\input{005conclusion}

\bibliography{main}
\bibliographystyle{icml2024}

\newpage
\appendix
\onecolumn
\input{appendix}

\end{document}

%% file: 000intro.tex
\section{Introduction}
\label{sec:intro}

Discovering relations among brain regions toward a specific phenotypic outcome from fMRI signals has been a crucial area in neuroimaging research. Reliable and informative relations help neuroscientists and clinical professionals to better understand brain functions, and thus improve clinical diagnosis and treatments \citep{kucian2008development,kucian2006impaired, li2015putting,li2015toward,satterthwaite2015linked,wang2016efficient}.
However, despite the clear multiplexity of the brain's involvement in cognition \citep{logue2014neural,barrasso2016neurobiological,knauff2010complex,reineberg2022context}, current imaging biomarker detection methods \citep{shen2017using,gao2019combining,li2021braingnn} focus only on the contributing roles of the pairwise connectivity edges. In contrast, most brain functions involve distributed patterns of interactions among multiple regions \citep{semedo2019cortical}. For instance, executive planning requires the appropriate communication of signals across many distinct cortical areas \citep{logue2014neural}. 
These high-order relationships, cannot always be decomposed into pairwise ones \cite{battiston2020networks, battiston2021physics, bick2023higher}, thus allowing them to capture information beyond the reach of pairwise ones \cite{do2020structural}. Consequently, relying solely on pairwise connectivity, without accounting for the brain's complex high-order structure, may result in inconsistent findings and low predictive performance across studies. 
Although recently there have been a few works \citep{zu2016identifying,xiao2019multi,li2022construction} working on discovering the high-order relationships of brain regions, they are unable to effectively extract meaningful patterns. This is because these methods usually first identify some candidates of high-order relationships and then perform feature selection. The candidates are obtained by enumerating only a small portion of all possible relations, or by clustering methods that are unrelated to the target. As a result, most informative high-order relations are likely excluded at the first stage. This inspires us to solve the problem in an end-to-end manner through a more expressive model and a more appropriate objective.

%1) the discovering process of high-order relationships is not guided by cognitive targets; 2) the weak expressiveness of traditional machine learning methods (e.g. lasso regression) hinders the accurate modeling of brain region interactions \citep{cao2022brain,richards2019deep}.

In this paper, we aim to identify high-order relationships that are informative towards a phenotypic outcome, such as a cognition score. However, unlike pairwise relations, the number of possible high-order relations is exponential. To identify the most informative ones from the exponential space, we propose our objective: \textit{maximally informative and minimally redundant} (MIMR). That is to say, we maximize the information contained in the high-order relationships towards a neurological outcome (\textit{informativeness}) while diminishing the participation of unrelated brain regions (\textit{redundant}).
% \todo{more citations in this paragraph to justify the principle would be awesome. We could even move some of the citations from previous paragraphs here if relevant.} 
Such a criterion, on the one hand, ensures the predictive performance of these high-order relationships; on the other hand, it endows the model with the capacity to identify more succinct and interpretable structures \citep{yu2020graph,miao2022interpretable,miao2022interpretable2,chen2023tempme}. A formal definition of the MIMR criterion could be found in Equation \ref{equ:ib2} from an information bottleneck point of view.

\begin{figure}
    \includegraphics[width=\linewidth]{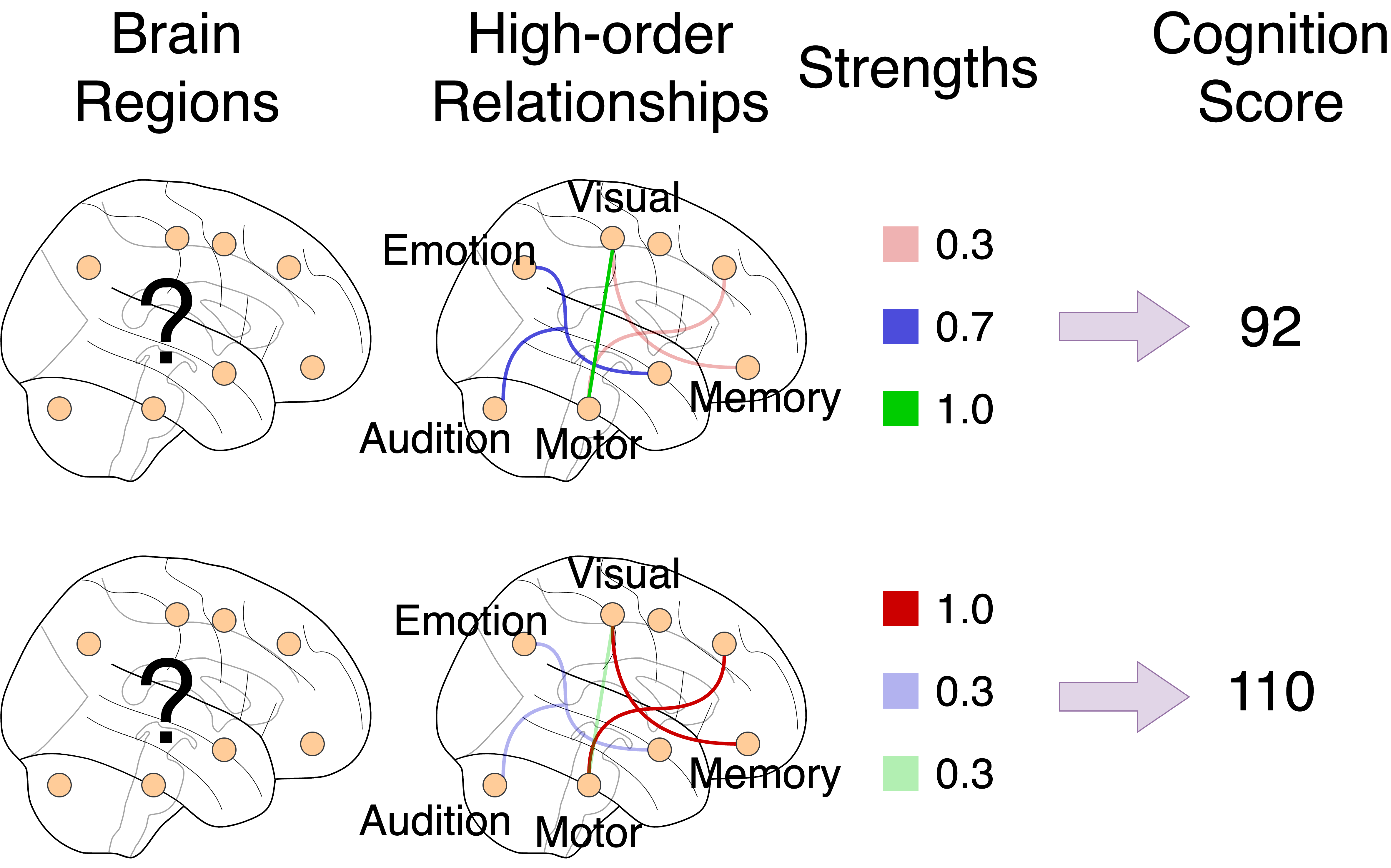}
    % \vskip 0.1in
    \caption{We identify high-order relationships of brain regions, where hyperedge structures and weights possess strong relevance to cognition (maximal informativeness). Meanwhile, they contain the least irrelevant information.}
    \label{fig:tasting}
% \vskip -0.1in
    
\end{figure}

% This property, on the one hand, ensures the high-order relationships encapsulate sufficient information related to cognition, which prevents the omission of any of them that makes a significant contribution to cognition. On the other hand, it allows us to diminish irrelevant ones, which will make the results as succinct as possible.

% \htl{we (propose to) formulate the capture of high-order relationship in the brain as the construction of the hypergraph, where the nodes are regions and the hyperedges connect all the related regions.}
% \htl{To better...(require a motivation of why MRMI), a desired hypergraph should maximize the correlation between hyperedges and the predicted property while diminishing the correlation with unrelated inputs, which is also demonstrated by xxx.}

% \htl{However, most of the previous works on hypergraph construction fail to xxx since }

We formulate high-order relationships as weighted hyperedges in a hypergraph, where regions are treated as nodes. Unlike a traditional graph where edges connect only two nodes, a hypergraph allows edges, known as hyperedges, to connect any number of nodes. The hypergraph should be weighted, and the weights of hyperedges are considered as strengths of high-order relationships, which contain the information relevant to the outcome (Figure \ref{fig:tasting}). 
% Therefore, our goal is to construct a weighted hypergraph model of each brain, where the nodes are regions and hyperedges are set to be identified. These hyperedges should be maximally informative and minimally redundant: \textit{informativeness} is the mutual information between the hyperedges and the prediction target (e.g. cognition score).
% We maximize it because we expect these hyperedges to carry more information about the target; \textit{redundancy} is the mutual information between the hyperedges and the inputs. We minimize it because we expect succinct relations. Acquiring such maximally informative and minimally redundant hyperedges is much more challenging than pairwise edges in a normal graph due to the exponential search space since the number of nodes connected by a hyperedge can be arbitrary.\htl{better to separate it into two sentences "...due to...since"}

% In addition, \stodo{perhaps ad
% isk of overfitting in the model.

However, current methods for hypergraph construction, which are mostly based on neighbor distances and neighbor reconstruction \citep{wang2015visual,liu2017elastic,jin2019robust,huang2009video}, are unsuitable in our context for several reasons: 1) they are unable to learn MIMR hyperedges due to the absence of a tractable objective for learning such hyperedges.  2) they fall short of learning consistent structures across subjects, which contradicts the belief that the cognitive function mechanism of healthy human beings should be similar \citep{wang2023inference}. 3) the number of hyperedges is restricted to the number of nodes, which may lead to sub-optimal performance. Furthermore, although information bottleneck (IB) has been a prevailing solution to learn MIMR representations in deep learning \citep{kim2021drop,alemi2016deep,luo2019significance}, existing IB methods focus on extracting compressed representations of inputs instead of identifying underlying structures such as hypergraphs. Harnessing the IB framework for identifying hypergraphs necessitates both architectural innovations and theoretical derivations.

\paragraph{Proposed Work} In this paper, we propose \textbf{Hy}pergraph of \textbf{B}rain \textbf{R}egions via mult\textbf{i}-head \textbf{D}rop-bottleneck (\mname{}), a novel approach for identifying maximally informative yet minimally redundant high-order relationships of brain regions. The overall pipeline of \mname{} is depicted in Figure \ref{fig:pipeline}. \mname{} is equipped with a \textsc{Constructor} and a \textsc{Weighter}. The \textsc{Constructor} identifies the hyperedge structures of brain regions by learning sets of masks, and the \textsc{Weighter} computes a weight for each hyperedge. To advance the IB principle for hyperedge identification, we further propose \textit{multi-head drop-bottleneck} and derive its optimization objective.

\mname{} avoids searching in an exponential space through learning masks to identify hyperedges, which guarantees efficiency. Its feature-agnostic masking mechanism design ensures \mname{} to learn consistent structures across subjects. Moreover, the model is equipped with a number of parallel heads, each of which is dedicated to a hyperedge. Through this, \mname{} is able to identify any number of hyperedges, depending on how many heads it is equipped with. Additionally, the proposed \multiheadname{} theoretically guarantees the maximal informativeness and minimal redundancy of the identified hyperedges. 

% To alleviate the data scarcity issue, \mname{} is equipped with a special architecture, where some components are shared by all datasets, and others are specific to the condition.

We evaluate our methods on the open-source ABIDE dataset and the restricted ABCD dataset. We quantitatively evaluate our approach by a commonly used protocol for studying brain connections, CPM~\citep{shen2017using} (Appendix \ref{sec:app_cpm}), and show that our model outperforms the state-of-the-art deep learning models by an average of $11.2\%$ on a comprehensive benchmark. Our post-hoc analysis demonstrates that hyperedges of higher degrees are considered more significant, which indicates the significance of high-order relationships in human brains.

%% file: 002preliminaries.tex
\section{Problem Definition \& Notations}
\label{sec:prelims}
% \subsection{Problem Definition \& Notations}
\paragraph{Input} 
 Our dataset is a collection of human subject's features and their phenotypic outcomes, which is represented by the pair $(X, Y)$ for each subject. $X \in \mathbb{R}^{N \times d}$ represents the features for each subject, where $N$ is the number of brain regions and $d$ is the feature size. Consistent with previous works \citep{kan2022brain,li2021braingnn}, the features are Pearson correlations derived from fMRI time series. $Y \in \mathbb{R}$ denotes the phenotypic outcome, such as the intelligent quotient. Section \ref{sec:experiments} and Appendix \ref{app:dataset} will elaborate more details about datasets and preprocessing procedures. 
 
 %For how to obtain features of regions from raw data, see for more details.

\paragraph{Goal}
Based on the input $X$, \mname{} aims to learn a weighted hypergraph of the brain, where regions are nodes. To achieve this, \mname{} identifies a collection of hyperedges $H = (\boldsymbol{h}^1, \boldsymbol{h}^2, \cdots, \boldsymbol{h}^K)$, and assigns weights $\boldsymbol{w} = [w^1, w^2, \cdots, w^K]^T$ for all hyperedges. These hyperedges and their weights, which represent strengths of hyperedges, are expected to be the most informative towards $Y$ yet the least redundant. 

\paragraph{Representation of Hyperedges}
As mentioned before, we use $H$ to denote the collection of hyperedge structures and $\boldsymbol{h}^k$ to denote the $k$-th hyperedge. To associate hyperedges with node memberships, we use the following representation for a hyperedge:
\begin{equation}
    \boldsymbol{h}^k = \boldsymbol{m}^k \odot X \in \mathbb{R}^{N\times d},
\end{equation}
where $\boldsymbol{m}^k \in \{0,1\}^N$ is a mask vector and $\odot$ denotes broadcasting element-wise multiplication. In other words, each $\boldsymbol{h}^k$ is a randomly row-zeroed version of $X$.

% \subsection{Connectome-based predictive modeling} 
% \label{sec:cpm}
% The connectome-based predictive modeling (CPM, \cite{shen2017using,finn2015functional}) has shown tremendous promise in recent years in detecting imaging biomarkers. \citep{rosenberg2015predicting, dubois2018resting, rosenberg2020functional, rosenberg2016neuromarker}. Such a model, based chiefly on functional MRI data, can measure the significance of the input edge weights, which is revealed by a correlation coefficient that reflects the correlation between the edge weights and the cognition score. One could expect a large correlation coefficient to indicate the high quality of edge weights. We utilize the CPM as an evaluation model to evaluate the quality of our learned hyperedges. More details of the CPM can be found in Appendix \ref{sec:app_cpm}.

%% file: 001related_work.tex
\section{Related Work}
\paragraph{Hypergraph Construction}
 Existing hypergraph construction methods are mostly based on neighbor reconstruction and neighbor distances. For example, the $k$ nearest neighbor-based method \citep{huang2009video} connects a centroid node and its $k$ nearest neighbors in the feature space to form a hyperedge. \citet{wang2015visual,liu2017elastic,jin2019robust,xiao2019multi} further refine these neighbor connections through various regularization. However, the number of hyperedges of these methods is restricted to the number of nodes, and hyperedges obtained by these methods are inconsistent across instances. 
 Cluster-based approaches learn community structure in a graph \cite{bannadabhavi2023community,ying2018hierarchical}, which can also be used to form hyperedges.
 \citet{zhang2022deep,zhang2018dynamic} proposed to iteratively refine a noisy hypergraph, which is obtained by the aforementioned methods. Therefore, they share the same limitations as the aforementioned methods. In addition, these methods are unable to learn MIMR hyperedges due to the absence of a tractable objective. Other methods, such as attributed-based methods~\citep{huang2015learning, joslyn2019high}, are ill-adapted to our context since they require discrete labels. 
 %or a prior graph topology, which is usually very noisy.
 Different from these methods, we provide a way to learn a consistent hypergraph through a deep-learning model. Furthermore, thanks to the proposed \multiheadname, these hyperedges are theoretically ensured MIMR.

\paragraph{High-Order Relationships in fMRI}

Although there are some methods working on high-order relationships in fMRI, they have limitations and are inconsistent with our MIMR objective.
\citet{xiao2019multi,li2022construction} used the existing non-learning-based hypergraph construction methods, which may lead to noisy and inexpressive hypergraphs. \citet{zu2016identifying,santoro2023higher} enumerated all hyperedges with degrees lower than 3, which can only discover a tiny portion of all possible hyperedges in exponential space and is not scalable to a large degree. \citet{rosas2019quantifying} proposed O-information, which reflects the balance between redundancy and synergy. The O-information metric is utilized by \citet{varley2023multivariate} to study fMRI data. However, the objective of these methods is not consistent with ours: although both of us are quantifying the redundancy of high-order relations, our method is to learn those that are most informative toward a cognition score, while theirs is to depict the synergy and redundancy within a system.

\paragraph{Information Bottleneck}
Information bottleneck (IB) \citep{tishby2000information} is a technique in data compression. The key idea is to extract a summary of data, which contains the most relevant information to the objective. \citet{alemi2016deep} first employed an IB view of deep learning. After that, IB has been widely used in deep learning. The applications span areas such as computer vision \citep{luo2019significance,peng2018variational}, reinforcement learning \citep{goyal2019infobot,igl2019generalization}, natural language processing \citep{wang2020learning} and graph learning \citep{yu2020graph,yu2022improving, xu2021infogcl,wu2020graph}. Unlike these studies that use IB to extract a compressed representation or a select set of features, our approach focuses on identifying the underlying structures of the data.

\paragraph{Connectivity-based Phenotypic Prediction}
Recently, deep learning techniques have been increasingly employed in predicting phenotypic outcomes based on the connectivity of brain regions. Most works \citep{ahmedt2021graph,li2019graph,cui2022interpretable,kan2022fbnetgen,cui2022braingb,said2023neurograph} model the brain network as a graph, in which regions act as nodes and pairwise correlations form the edges. These methods predominantly utilize Graph Neural Networks (GNNs) to capture the connectivity information for predictions. In addition to GNNs, \citet{kan2022brain} proposed to use transformers with a specially designed readout module, leveraging multi-head attention mechanisms to capture pairwise connectivity. However, all of these methods heavily rely on pairwise connectivity and neglect more intricate higher-order relationships. This oversight, on the one hand, leads to sub-optimal prediction performances and, on the other hand, prevents domain experts from acquiring insightful neuroscience interpretations, given that brain functions often involves multiple regions.

%% file: 003approach.tex
\section{Method}
\label{sec:approach}
\paragraph{Method Overview}

\mname{} consists of a \textsc{Constructor} $\mathcal{F}_c$, a \textsc{Weighter} $\mathcal{F}_w$, and a \textsc{LinearHead} $\mathcal{F}_l$. At a high level, the \textsc{Constructor} $\mathcal{F}_c$ is responsible for identifying hyperedges $H$ from the data to construct the hypergraph. After that, the \textsc{Weighter} $\mathcal{F}_w$ calculates a weight for each hyperedge. Finally, based on all the weights $\boldsymbol{w}$, the \textsc{LinearHead} $\mathcal{F}_l$ predicts the label $Y$.
An illustration of this pipeline is shown in Figure \ref{fig:pipeline}. The pipeline can be formulated as
\begin{equation}
% \vspace{-1em}
    \label{equ:pipeline}
    X \xrightarrow[\mathcal{F}_c]{} H \xrightarrow[\mathcal{F}_w]{} \boldsymbol{w}\xrightarrow[\mathcal{F}_l]{} Y.
\end{equation}

We will elaborate on the details of the architecture below.

\begin{figure*}
\vspace{0.8em}
    \centering
    \includegraphics[width=\linewidth]{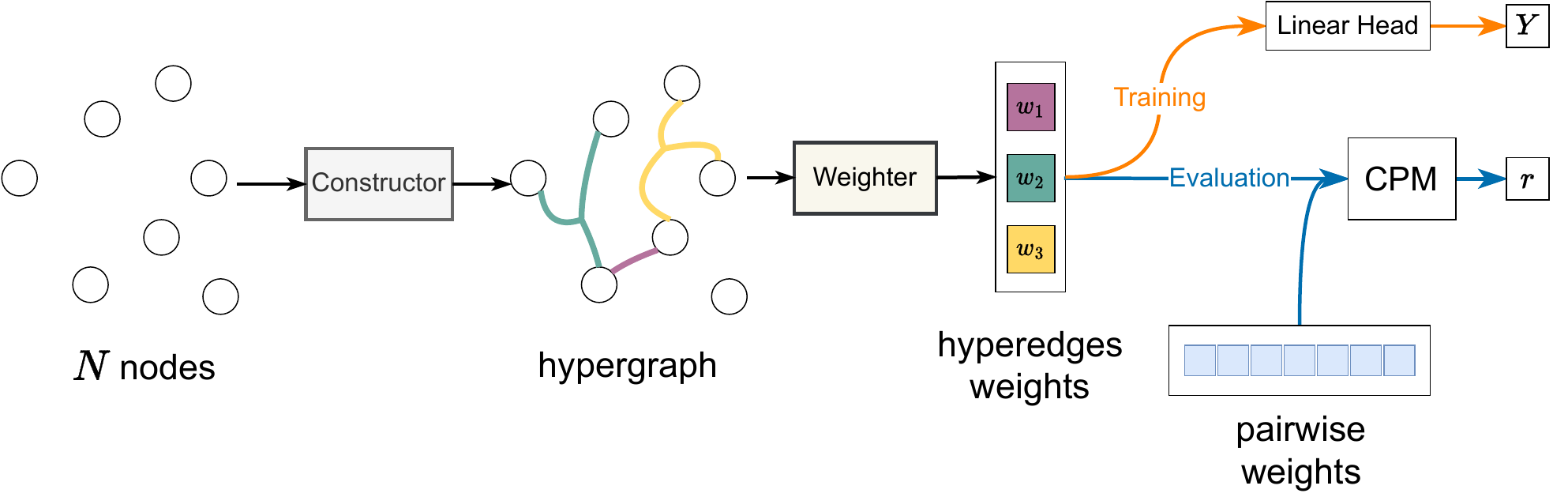}
    % \vskip 0.1in
    \caption{Overview of the \mname{} pipeline when the total number of hyperedges $K=3$. Hyperedge are in distinct colors for clarity.  The \textsc{Constructor} identifies hyperedges in the hypergraph, where regions are nodes. The \textsc{Weighter} computes a weight for each hyperedge. These weights, representing strengths of hyperedges, are expected to be informative in terms of our target $Y$. There are two separate phases after obtaining weights of hyperedges: 1) Training. The model's parameters are trained under the supervision of $Y$; 2) Evaluation. The output weights, as well as pairwise weights, are fed into the CPM (see Appendix \ref{sec:app_cpm}).}
    \label{fig:pipeline}
    % \vskip -0.1in
\end{figure*}

\subsection{Learning the Hypergraph by Multi-head Masking}
\label{sec:learn-by-mask}

Each instance (human subject) is represented by $X = [X_1, X_2, \cdots, X_N]^T \in \mathbb{R}^{N\times d}$, where $X_i \in \mathbb{R^d}$ is a column vector representing the features of region $i$. These regions are nodes in the hypergraph we are going to construct. Hyperedges in the hypergraph can be beneficial in the learning process below because it is essential to model the relationships between more than two regions.

\paragraph{Hyperedges Construction}
In this paragraph, we elaborate on how the \textsc{Constructor} identifies the hyperedges, i.e. $H=\mathcal{F}_c(X)$.

Suppose the number of hyperedges is $K$, which is a predefined hyperparameter. We assign a head to each hyperedge. Each head is responsible for constructing a hyperedge by selecting nodes belonging to that hyperedge. 

Specifically, to construct the $k$-th hyperedge, the \textsc{Constructor}'s $k$-th head outputs a column vector $\boldsymbol{m}^k \in \{0,1\}^N$, where each element in the vector corresponds to a brain region,

\begin{equation}
\label{equ:shallow}
% \begin{aligned}
\boldsymbol{m}^k = [\mathbbm{1}(p^k_{\theta,1}), \mathbbm{1}(p^k_{\theta,2}), \cdots, \mathbbm{1}(p^k_{\theta,N})]^T \in \{0, 1\}^N,
% \end{aligned}
\end{equation}
where $p^k_{\theta, i} \in [0,1], i=1,2,\cdots, N$ are learnable probabilities. $\mathbbm{1}: [0,1] \mapsto \{0,1\} $ is an indicator function, which is defined as $\mathbbm{1}(x) = 1$ if $x > 0.5$ and $\mathbbm{1}(x) = 0$ if $x \leq 0.5$.  And $\boldsymbol{m}^k$ is a column vector corresponding to the $k$-th hyperedge. Note that since there is no gradient defined for the indicator operation, we employ the stop-gradient technique \citep{oord2017neural,bengio2013estimating} to approximate it.

\begin{figure}[t]
    \centering
    \includegraphics[width=\linewidth]{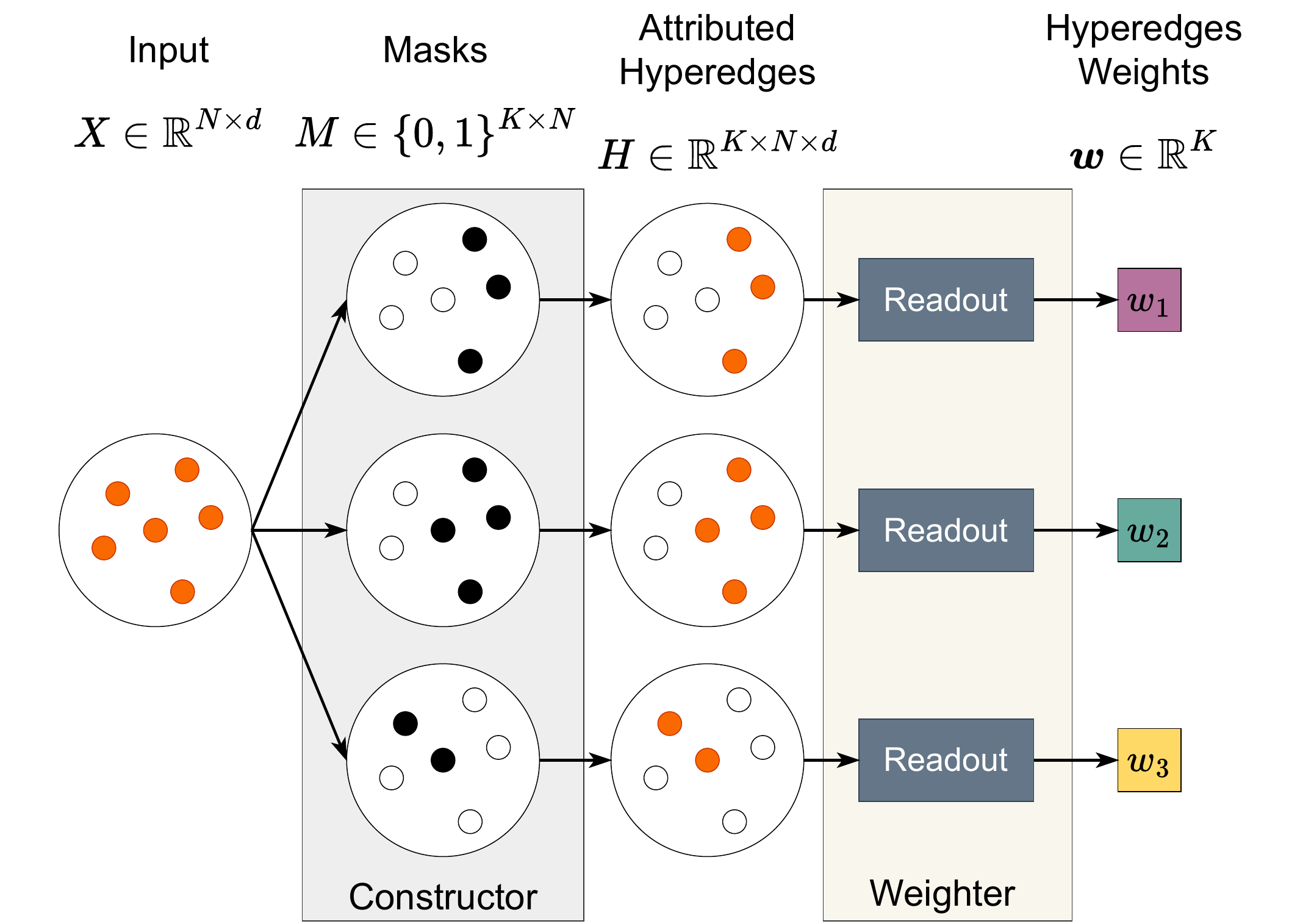}
    % \vskip 0.1in
    \caption{Architecture details of the \textsc{Constructor} and the \textsc{Weighter}  when the number of nodes $N=6$ and the number of hyperedges $K=3$. At a high level, the \textsc{Constructor} learns the hyperedge structure by masking nodes. The \textsc{Weighter} computes the weight of each hyperedge based on the hyperedge's member nodes and their features.}
    \label{fig:enter-label}
    % \vskip -0.1in
    \vspace{-2em}
\end{figure}

In the vector $\boldsymbol{m}^k$, $0$ indicates nodes that are masked out, and $1$ indicates nodes that are not masked. Nodes that are not masked are considered to form a hyperedge together. We use $\boldsymbol{h}^k$ to represent the masked version of $X$
\begin{equation}
\begin{aligned}
    \boldsymbol{h}^k &= \boldsymbol{m}^k \odot X \\&= [\boldsymbol{m}^k_1 X_1, \boldsymbol{m}^k_2 X_2, \cdots, \boldsymbol{m}^k_N X_N] \in \mathbb{R}^{N\times d},   
\end{aligned}
\end{equation}
where $\odot$ is the broadcast element-wise multiplication. $\boldsymbol{m}^k_j$ is the $j$-th element of the vector $\boldsymbol{m}^k$.  

We obtain $K$ hyperedges for $K$ sets of masks. We use $H$ to denote the collection of all hyperedges.
\begin{equation}
    H = (\boldsymbol{h}^1, \boldsymbol{h}^2, \cdots, \boldsymbol{h}^K).
\end{equation} 

\paragraph{Hyperedge Weighting}
After obtaining the structure (i.e. member nodes) of each hyperedge,  the \textsc{Weighter} will calculate each hyperedge's weight, which is supposed to indicate the importance of that hyperedge, based on the member nodes and their features, i.e. $\boldsymbol{w} = \mathcal{F}_w(H)$.

These weights are obtained by a $\operatorname{Readout}$ module, which is composed of: 1) summing over all the non-masked nodes feature-wisely; 2) dim reduction operation.

\begin{equation}
\label{equ:readout}
    w^k = \operatorname{Readout}(\boldsymbol{h}^k) = \operatorname{DimReduction}({\boldsymbol{m}^k}^T\boldsymbol{h}^k) \in \mathbb{R},
\end{equation}

where $w^k$ is the weight of the $k$-th hyperedge and $\operatorname{DimReduction}$ is an MLP with ReLU activations, where the output dimension is $1$.
For all hyperedges, we obtain $K$ hyperedges in total, $\boldsymbol{w} = [w^1, w^2, \cdots, w^K]^T \in \mathbb{R}^K$. Finally, these weights will be fed into the final linear head to predict the label of the instance, 
\begin{equation}
\label{equ:linearhead}
    \hat{Y} = \mathcal{F}_l(\boldsymbol{w}) \in \mathbb{R}.
\end{equation}

In contrast to previous hypergraph construction methods \citep{jin2019robust,xiao2019multi}, which identify hyperedges by refining neighbors and simply aggregating node features, \mname{} makes these procedures learnable and thus is able to identify MIMR hyperedges in a data-driven way through expressive neural networks. The number of hyperedges $K$ is decided according to the study in Appendix \ref{app:k_choice}.

\paragraph{Computational Complexity}
The computational complexity of our model is $O(N^2K)$, which is just the same scale as that of MLPs even though we are addressing a more challenging task: identifying high-order relationships in an exponential space. Details of the complexity calculation can be found in Appendix \ref{app:complexity}.

\subsection{Optimization Framework}
 Since there are no existing IB frameworks that can be applied in our context, we propose a new IB framework named \multiheadname{} to optimize \mname{}. To adopt an information bottleneck view of \mname{}, we consider $X$, $Y$ and $H$ are random variables in the Markovian chain $X \leftrightarrow Y \leftrightarrow H$. According to our MIMR objective, we optimize
\begin{equation}
\label{equ:ib2}
    \arg \max I(H; Y) - \beta I(H; X),
\end{equation}
where $I(\cdot ; \cdot)$ denotes the mutual information. $I(H; Y)$ corresponds to the informativeness and $I(H; X)$ corresponds to the redundancy. $\beta$ is a coefficient trading off informativeness and redundancy. Since optimizing the mutual information for high-dimensional continuous variables is intractable, we instead optimize the lower bound of Equation \ref{equ:ib2}. Specifically, for the first term (informativeness), it is easy to show
\begin{equation}
\label{equ:lowerbound}
\begin{aligned}
    I(H; Y) &= \mathbb{H}[Y] -  \mathbb{H}[{Y|H}] \\
    &=\mathbb{H}[Y] + \mathbb{E}_{p(Y, H)}[\log p(Y|H)] \\
    & = \mathbb{H}[Y] + \mathbb{E}_{p(Y, H)}[\log q_\phi(Y|H)] \\
    &\quad + \mathbb{E}_{p(H)}[\operatorname{KL}(p(Y|H)|q_\phi(Y|H))] \\
   & \geq \mathbb{H}[Y] + \mathbb{E}_{p(Y, H)}[\log q_\phi(Y|H)],
\end{aligned}
\end{equation}
where $\mathbb{H}[\cdot]$ is the entropy computation.
Since there is no learnable component in the entropy of $Y$, we only need to optimize the second term $\mathbb{E}_{p(Y, H)}[\log q_\phi(Y|H)]$. $q_\phi$ can be considered as a model that predicts $Y$ based on $H$, which essentially corresponds to $\mathcal{F}_l \circ \mathcal{F}_w$, where $\circ$ is the function composition. In practice, we set $q_\phi$ as a Gaussian model with variance $1$ as most probabilistic machine learning models do for continuous data modeling \citep{alemi2016deep,luo2019significance,peng2018variational,kingma2013auto}.

For the second term (redundancy) in Equation \ref{equ:ib2}, we have

\begin{proposition}
    (Upper bound of $I(H;X)$ in \multiheadname{})
\begin{equation}
\label{equ:multi_drop}
    \begin{aligned}
        I(H; X) 
        \leq \sum_{k=1}^K
        I(\boldsymbol{h}^k; X)
        &\leq \sum_{k=1}^K \sum_{i=1}^N
        I(\boldsymbol{h}^k_i; X_i)\\
        &= \sum_{k=1}^K\sum_{i=1}^N \mathbb{H}[X_i] (1 - p^k_{\theta, i}),
    \end{aligned}
\end{equation}
\end{proposition}

where $\boldsymbol{h}^k_i$ and $X_i$  is the $i$-th row of $\boldsymbol{h}^k$ and $X$ respectively.  $p^k_{\theta,i}$ is the mask probability in Equation \ref{equ:shallow}. $\mathbb{H}$ is the entropy computation. The equality holds if and only if nodes are independent and hyperedges do not overlap. It's important to clarify that optimizing Equation \ref{equ:multi_drop} does not imply a penalty for overlaps. The second inequality is inspired by \cite{kim2021drop}. The proof of the proposition can be found in Appendix~\ref{app:proof1}.

% \begin{equation}
%     \mathcal{F}_c(X) =  \mathbbm{1}(p_\theta) \odot X
% \end{equation}
% \begin{equation}
%     q_\phi (H) =\mathcal{F}_l \circ \mathcal{F}_w (H)
% \end{equation}
Therefore, instead of optimizing the intractable objective Equations \ref{equ:ib2}, we optimize its upper bound (i.e. loss function) according to Equations \ref{equ:lowerbound} and \ref{equ:multi_drop}.
\begin{equation}
\begin{aligned}
    \mathcal{L} &= \| Y - \mathcal{F}_l \circ \mathcal{F}_w \circ \mathcal{F}_c (X) \|_2^2 + \beta \sum_{k=1}^K\sum_{i=1}^N \mathbb{H}[X_i] (1 - p^k_{\theta,i}) \\&\geq - I(H, Y) + \beta I(H, X).
\end{aligned}
\end{equation}
In conclusion, The learnable components are the shallow embeddings in Equation \ref{equ:shallow}, the $\operatorname{DimReduction}$ MLP in Equation \ref{equ:readout} and the \textsc{LinearHead} $\mathcal{F}_l$ in Equation \ref{equ:linearhead}. For how we choose the trade-off coefficient $\beta$, see Appendix \ref{app:beta} for more discussions.

%% file: 004experiments.tex
\section{Experiments}
\label{sec:experiments}

In this section, we conduct experiments to validate the quality of the learned hyperedges in terms of the predictive performance towards the cognition phenotype outcome. Furthermore, we conduct ablation studies to validate the key components in our model. We also analyze our results both quantitatively and qualitatively.

\subsection{Predictive Performance of Hypereges}
% \begin{algorithm}
% \begin{algorithmic}
% \STATE    $d \gets \text{an random integer between } 2 \text{ and } 30. $

% \STATE $\mathcal{H} \gets \{\}$
% \FOR{$i = 1:d$}
% \STATE $h \gets \text{randomly sample } d \text{ nodes}$
% \STATE $\mathcal{H} \gets \mathcal{H} \cup \{h\}$
% \ENDFOR

% \FORALL{$h$ in $H$}
% \FORALL{node in $h$}

% \STATE $w \gets U(0, 1)$
% \STATE $x \gets U(0, 2w)$
% \STATE $y \gets y + $

% \ENDFOR
% \ENDFOR

% \end{algorithmic}

% \caption{The procedure to construct the synthetic dataset.}
% \label{alg:syn}
% \end{algorithm}

% \paragraph{Real-World Datasets}
\paragraph{Datasets}
\label{par:real_dataset}

We consider two fMRI datasets: 

1) \textit{Autism Brain Imaging Data Exchange (ABIDE)} \citep{craddock2013neuro} is an open-source dataset. This dataset involved resting-state fMRI of patients from 17 international sites, as well as the anatomical and phenotypic data. Regions are obtained by Craddock 200 atlas \citep{craddock2012whole}. We use the preprocessed version from the official website. For prediction targets, we choose three intelligence quotients: FIQ (full-scale intelligence), VIQ (verbal intelligence quotient), and PIQ (performance intelligence). 
%We aim to show our model's effectiveness on different cognition targets.

2) \textit{Adolescent Brain Cognitive Development (ABCD)} \citep{casey2018adolescent} is one of the largest public fMRI datasets. Access is limited and requires adherence to a rigorous data request procedure to acquire the data. The data is collected from $11,875$ children aged between 9 to 10 years old. The functional MRI (fMRI) data is collected from children when they were resting and when they performed three tasks (SST, EN-back, MID). We use the ABCD imaging data collected from the baseline (release 2.0) as well as the 2-year follow-up (release 3.0). In conclusion, we obtain $8$ sub-datasets (we refer to them as \textit{datasets} from now on) from $2$ timepoints under $4$ tasks. Regions are obtained by AAL3v1 atlas \citep{rolls2020automated}. For preprocess procedures of the ABCD dataset, please refer to Appendix \ref{app:dataset} for more details. For the prediction target, we consider fluid intelligence as our label. Fluid intelligence reflects the general mental ability and plays a fundamental role in various cognitive functions.
%We aim to show our model's effectiveness on different cognition targets
%Recent literature has seen success in predicting intelligence based on pairwise predictive modeling \citep{dubois2018resting,dubois2018distributed}, and we aim to improve the prediction accuracy of the current methods by involving high-order relationships on both resting-state and task-based data.

For region features, consistent with previous connectivity-based methods \citep{li2021braingnn,kan2022brain,ktena2018metric,said2023neurograph}, we use a region's Pearson correlation coefficients to all other regions as the region features.
Other details of the data preprocessing and statistics of each dataset are summarized in Appendix \ref{app:dataset}.

% \paragraph{Features}
% \paragraph{Data preprocessing}
% \label{par:prepocessing}
% The raw fMRI data of an instance is represented in four dimensions (3 spatial dimensions + 1 temporal dimension), which can be imagined as a temporal sequence of 3D images. First, brain images are parceled into regions (or nodes) using the AAL3v1 atlas \citep{rolls2020automated}. Following previous works \citep{kan2022brain,li2021braingnn,thomas2022self}, each region's time series is obtained by averaging all voxels in that region. Consistent with previous connectivity-based methods \citep{li2021braingnn,kan2022brain,ktena2018metric,said2023neurograph}, for each region, we use its Pearson correlation coefficients to all regions as its features. We randomly split the data into train, validation, and test sets in a stratified fashion. The split ratio is 8:1:1.

\paragraph{Evaluation Metric}
To evaluate the quality of hyperedges obtained by \mname, we use CPM~\citep{shen2017using}, a standard model that could evaluate the relevance between the connectivity and the prediction target, due to its high impact in the community. 
In the original implementation of CPM, weights of pairwise edges are obtained by Pearson correlation between nodes. These weights, as pairwise connectivity, are fed into the CPM. CPM will output a metric that measures the overall correlation between edges and the prediction target, which can be considered as a measure of edge qualities. This process is formulated as 
\begin{equation}
\label{equ:metric0}
    r^\prime = \operatorname{CPM}(\boldsymbol{w}_p, Y),
\end{equation}
where $\boldsymbol{w}_p \in \mathbb{R}^{K_p}$ denotes the pairwise edge weights and $K_p$ is the total number of pairwise edges. $r^\prime$ is a metric that measures the quality of weights based on positive and negative correlations to the outcome. 

To evaluate the quality of the learned weights for our model, we replace the pairwise edge weights with the learned high-order weights $\boldsymbol{w}_h \in \mathbb{R}^{K_h}$, and thus adjust Equation \ref{equ:metric0} to 
\begin{equation}
\label{equ:metric}
    r = \operatorname{CPM}(\boldsymbol{w}_h, Y),
\end{equation}

% where $[\cdot \| \cdot]$ denotes the concatenation of two vectors. 
Comparing $r$ to $r^\prime$ reflects the quality of learned weights in terms of the prediction performance since it measures the overall correlation between weights and the prediction target. In our model, $\boldsymbol{w}_h = \boldsymbol{w}$, which is the learned hyperedge weights.

\begin{table*}
\caption{$r$ values of our hyperedges compared to baselines on the ABCD dataset. Results are averaged over 10 runs. Deterministic methods do not have standard deviations.}
\label{tab:main}
% \vskip 0.1in
\resizebox{\linewidth}{!}{
\begin{tabular}{ccllllllll}
    \toprule
{Type} & {Model}& {SST 1} & {EN-back 1} & {MID 1}  & {Rest 1}& {SST 2} & {EN-back 2} & {MID 2} & {Rest 2} \\
         \midrule
Standard &pairwise& $0.113$ & $0.218$ &  $0.099$ & $0.164$  & $0.201$&$0.322$&$0.299$ & $0.289$ \\
\shline
\multirow{3}{*}{\shortstack{Hypergraph\\Construction}}
&$k$NN &  $0.115$ &  $0.268$ & $0.168$ &  $0.127$  & $0.257$&$0.266$&$0.238$ & $0.315$\\
&$l_1$ hypergraph  & $0.099$ &$0.223$  &  $0.125$ & $0.126$ &$0.145$&$0.295$&$0.242$& $0.259$\\
&$l_2$ hypergraph &$0.096_{\pm0.002}$  & $0.197_{\pm0.003}$ &  $0.118_{\pm0.003}$ & $0.157_{\pm0.016}$ &$0.203_{\pm0.005}$&$0.272_{\pm0.004}$&$0.289_{\pm0.011}$&$0.307_{\pm0.006}$\\
\shline
\multirow{4}{*}{\shortstack{Connectivity\\based\\Prediction}}
&BrainNetGNN &  $0.227_{\pm0.060}$ & $0.287_{\pm0.043}$ &$0.266_{\pm0.046}$ & $0.221_{\pm0.040}$ &$0.468_{\pm0.058}$&$0.480_{\pm0.068}$& $0.506_{\pm0.057}$&$0.453_{\pm0.028}$\\
&BrainGB & $0.190_{\pm0.073}$ & $0.214_{\pm0.051}$ &$0.265_{\pm0.048}$& $0.176_{\pm0.066}$ & $0.447_{\pm0.089}$ & $0.483_{\pm0.077}$ &$0.458_{\pm0.064}$ & $0.432_{\pm0.076}$\\
&BrainGNN & $0.262_{\pm0.030}$ & $0.235_{\pm0.032}$ &  $0.260_{\pm0.049}$ & $0.185_{\pm0.058}$ & $0.455_{\pm0.028}$ & $0.391_{\pm0.077}$&$0.445_{\pm0.078}$ &$0.368_{\pm0.041}$\\
&BrainNetTF &$\underline{0.327_{\pm 0.084}}$ & $\underline{0.338_{\pm 0.056}}$  & $\underline{0.370_{\pm 0.098}}$ & $\mathbf{0.334_{\pm 0.084}}$ & $\underline{0.633_{\pm0.178}}$& $\underline{0.631_{\pm0.142}}$ & $\underline{0.629_{\pm0.123}}$ & $\underline{0.588_{\pm0.138}}$ \\ 
\shline
Ours  & \mname & $\mathbf{0.361_{\pm0.058}}$  & $\mathbf{0.348_{\pm0.061}}$ & $\mathbf{0.386_{\pm0.060}}$ &  $\underline{0.223_{\pm0.056}}$ &$\mathbf{0.738_{\pm0.054}}$ & $\mathbf{0.714_{\pm0.037}}$&$\mathbf{0.816_{\pm0.053}}$ &$\mathbf{0.730_{\pm0.049}}$\\
         \bottomrule
    \end{tabular}}
% \vspace{-0.5em}
\end{table*}

\begin{table}[t]
 \caption{$r$ values of our hyperedges compared to baselines on the ABIDE dataset. Results are averaged over 10 runs. Deterministic methods do not have standard deviations.}
    \label{tab:abide}
% \vskip 0.1in
\centering
    \resizebox{\linewidth}{!}{
    \begin{tabular}{cclll}
    \toprule
      {Type} & {Model} & {FIQ} & {VIQ} & {PIQ}  \\
      \midrule
        Standard & pairwise & $0.052$ & $0.124$ & $0.056$ \\
        \shline
        \multirow{3}{*}{\shortstack{Hypergraph\\Construction}}
        & $k$NN & $0.023$ & $0.093$ & $0.056$ \\
        & $l_1$ hypergraph & $0.043$ & $0.125$ & $0.061$ \\
        & $l_2$ hypergraph & $0.148_{\pm0.000}$ & $0.141_{\pm0.014}$ & $0.063_{\pm0.004}$ \\
        \shline
        \multirow{4}{*}{\shortstack{Connectivity\\based\\Prediction}}
        &BrainNetGNN & $\underline{0.162_{\pm0.042}}$ & $\underline{0.199_{\pm0.042}}$ & $\underline{0.223_{\pm0.025}}$ \\
        &BrainGB & $0.125_{\pm0.119}$ & $0.154_{\pm0.068}$ & $0.157_{\pm0.053}$ \\
        &BrainGNN & $0.105_{\pm0.041}$ & $0.176_{\pm0.049}$ & $0.159_{\pm0.051}$ \\
        &BrainNetTF & $0.132_{\pm0.111}$ & $0.176_{\pm0.053}$ & $0.180_{\pm0.054}$ \\ 
        \shline
        Ours  & \mname & $\mathbf{0.181_{\pm0.040}}$ & $\mathbf{0.204_{\pm0.031}}$ & $\mathbf{0.245_{\pm0.042}}$\\
     \bottomrule
    \end{tabular}
    }
\end{table}

\begin{table*}[t]
\caption{Ablation studies on the masking mechanism. Results are averaged over 10 runs.}
% \vskip 0.1in
\label{tab:ablation_mask}
\resizebox{\linewidth}{!}{
\begin{tabular}{cccccccccccc}
    \toprule
{Model} & {SST 1} & {EN-back 1} & {MID 1} &{Rest 1}& {SST 2} & {EN-back 2} & {MID 2}  &{Rest 2}\\
         \midrule
 \mname  & $\mathbf{0.361_{\pm0.058}}$ &$\mathbf{0.348_{\pm0.061}}$ & $\mathbf{0.386_{\pm0.060}}$ & $\underline{0.223_{\pm0.056}}$ &$\mathbf{0.738_{\pm0.054}}$& $\mathbf{0.714_{\pm0.037}}$ &  $\mathbf{0.816_{\pm0.053}}$ & $\mathbf{0.730_{\pm0.049}}$\\
 \mname$_\mathrm{NoMask}$  & $0.297_{\pm0.035}$ & $0.274_{\pm0.057}$ & $\underline{0.323_{\pm0.059}}$& $0.221_{\pm0.034}$&$0.653_{\pm0.036}$&$0.599_{\pm0.059}$ &$0.757_{\pm0.021}$ &$0.543_{\pm0.038}$\\
 \mname$_\mathrm{RndMask}$  & $0.256_{\pm0.069}$ & $0.191_{\pm0.046}$ & $0.255_{\pm0.080}$ & $0.190_{\pm0.051}$ &$0.541_{\pm0.069}$& $0.514_{\pm0.038}$&$0.598_{\pm0.064}$ &$0.482_{\pm0.083}$\\
 \mname$_\mathrm{SoftMask}$  & $\underline{0.343_{\pm0.042}}$ & $\underline{0.314_{\pm0.040}}$ & $0.320_{\pm0.055}$& $\mathbf{0.245_{\pm0.061}}$ &$\underline{0.707_{\pm0.042}}$ &$\underline{0.662_{\pm0.058}}$&$\underline{0.796_{\pm0.031}}$ & $\underline{0.655_{\pm0.030}}$ \\
         \bottomrule
    \end{tabular}}
    % \vspace{-0.6em}
\end{table*}

\paragraph{Baselines}
We compare our method with $3$ classes of baselines: 1) \textit{standard} method, which is exactly the classical method that predicts outcomes based on pairwise edges \citep{shen2017using, DADI2019115, bhaa407}. The comparison with standard methods shows whether the high-order connectivity has its advantage over the classical pairwise one or not. 
2) \textit{hypergraph construction} methods. We consider $k$NN \citep{huang2009video}, $l_1$ hypergraph \citep{wang2015visual}, and $l_2$ hypergraph \citep{jin2019robust}.
3) \textit{connectivity-based  phenotypic prediction} methods, which are state-of-the-art predictive models based on brain connectivity. We consider BrainNetGNN \citep{mahmood2021deep}, BrainGNN \citep{li2021braingnn}, and BrainNetTF \citep{kan2022brain}. BrainGB \citep{cui2022braingb} is a study of different brain graph neural network designs and we include its best design as a baseline. Since none of these models are able to identify hyperedge structures of brain regions, we input their last layer embeddings (each entry as a weight) into the CPM model. Note that our weights $\boldsymbol{w}$ are also last layer embeddings in \mname{}.

\paragraph{Implementation \& Training Details}
Hyperparameter choices and other details can be found in Appendix \ref{app:training}.

% \subsection{Results}

% \paragraph{Quality of Hyperedges}
\paragraph{Results}
We report $r$ values by CPM in Table \ref{tab:abide} and Table \ref{tab:main}. As we can see, on the ABIDE dataset, \mname{} consistently outperforms all the baselines on different targets, with an average improvement of $8.8\%$ compared to the state-of-the-art model.
On the ABCD dataset, \mname{} outperforms the state-of-the-art predictive models on $7$ datasets of ABCD, with an average improvement of \avgimp{}. The results demonstrate our model is able to learn informative hyperedges towards different phenotypic outcomes from fMRI data of various brain states. Rest 1 is the only dataset that our model reaches the second-best. We conduct analyses and propose the potential reasons in Appendix~\ref{app:why_fail}.

Further, the comparison between our model and the pairwise baseline demonstrates the superiority of incorporating high-order relationships over relying solely on pairwise ones.

\paragraph{Runtime} \mname{} outperforms all other deep learning baselines in efficiency, with $87\%$ faster than the second-fastest one (BrainNetTF). Refer to Appendix \ref{app:runtime} for more runtime details.
% \rex{more in terms of result analysis. which dataset benefits the most? reviewers sometimes like to ask: what about the Rest1 task where the proposed method performs worse? anything special about it?}

\paragraph{Ablation Studies}
We conduct an ablation study on the effect of our masking mechanism. Specifically, we compare our model with $3$ variants: 1) \mname$_\mathrm{RndMask}$: Replace the learnable masks with random masks with the same sparsity, initialized at the beginning of training. 2) \mname$_\mathrm{NoMask}$: Do not mask at all, which means all nodes and their features are visible to each head. 3) \mname$_\mathrm{SoftMask}$: Remove the indicator function and use $p^k_{\theta, i}$ directly in Equation \ref{equ:shallow}. Ablation results are shown in Table \ref{tab:ablation_mask}. We find the original \mname{} and the \mname$_\mathrm{SoftMask}$ outperform all other variants, which demonstrates the effect of learnable masks. Moreover, the original \mname{} is better than its soft version \mname$_\mathrm{SoftMask}$, which demonstrates our sparse and succinct representations preserve better information than smooth ones. Other ablation studies such as the choices of the number of hyperedges and choices of $\beta$ can be found in Appendix \ref{app:beta}.

\begin{figure*}
% \vspace{-0.5em}
\begin{center}
    \begin{subfigure}[t]{0.32\textwidth}
        \includegraphics[width=\textwidth]{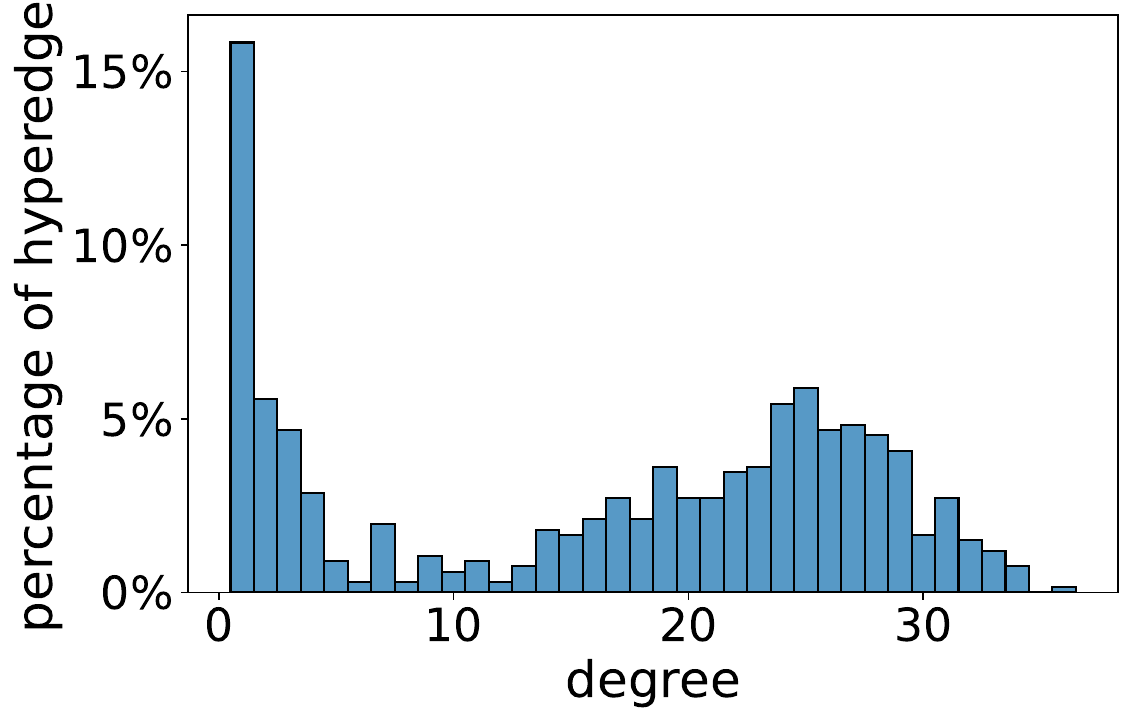}
        \caption{}
        \label{fig:degree_dist}
    \end{subfigure}
    \hspace{1em}
    \begin{subfigure}[t]{0.27\textwidth}
        \includegraphics[width=\textwidth]{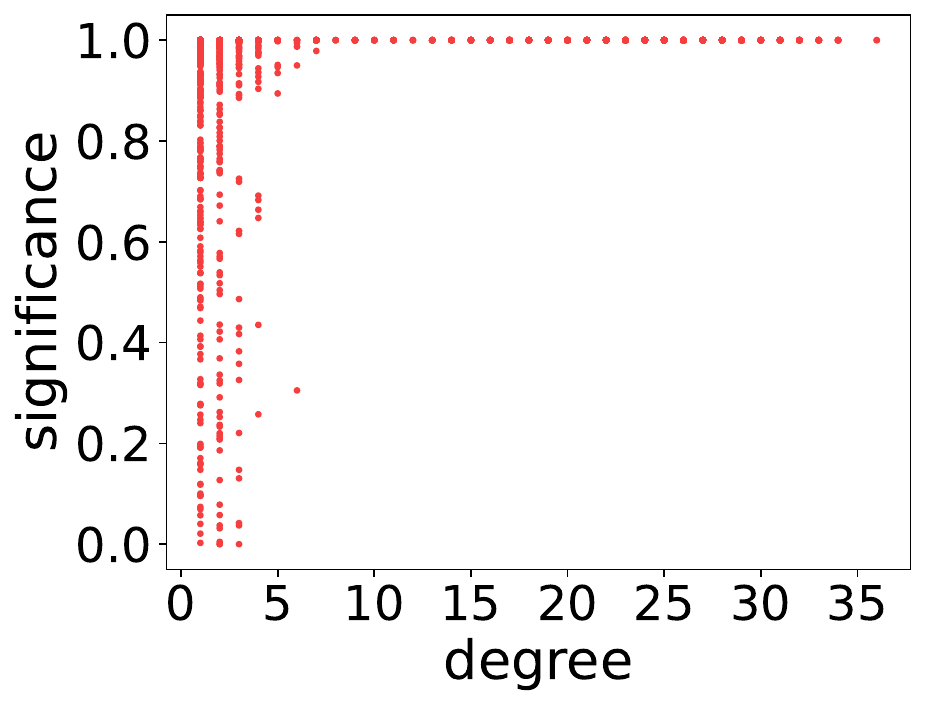}
        \caption{}
        \label{fig:degree_sig}
    \end{subfigure}
        \hspace{1em}
    \begin{subfigure}[t]{0.28\textwidth}
        \includegraphics[width=\textwidth]{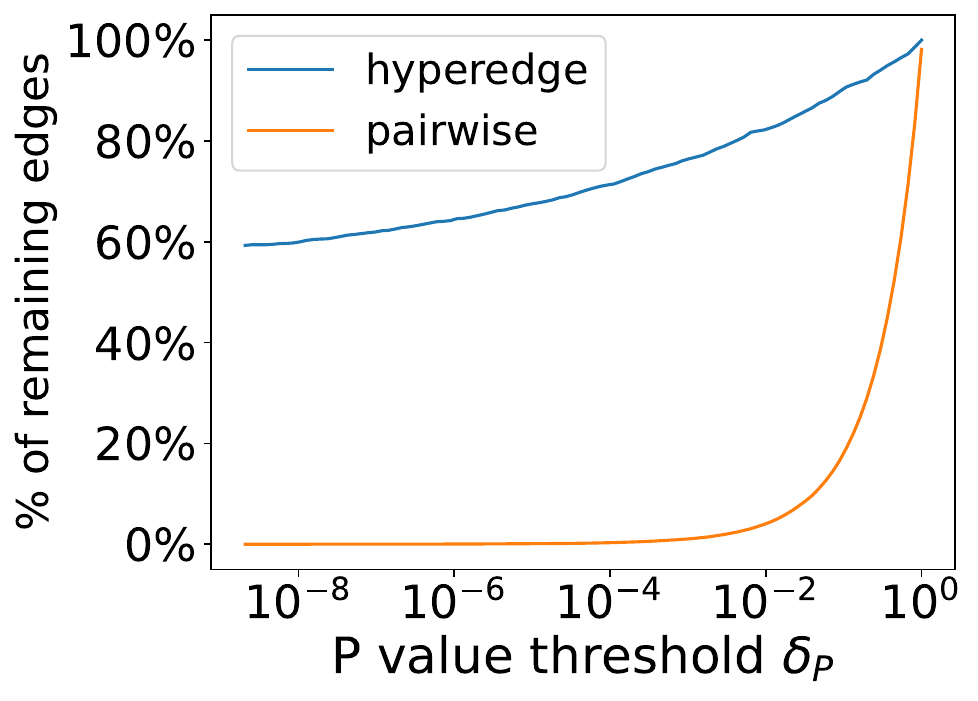}
        \caption{}
        \label{fig:cmp_sig_pair}
    \end{subfigure}
    % \vskip 0.1in
    \caption{Hyperedge profiles. \textbf{(a)} Hyperedge degree distribution of learned hyperedges. \textbf{(b)} Correlation between hyperedge degree and significance. \textbf{(c)} Comparison between the number of hyperedges and pairwise edges under different significance thresholds. The total number of hyperedges is $32$. And the total number of pairwise edges is $26,896$.}
    \label{fig:hyperedge_profile}
    % \vskip -0.1in
    \end{center}
\end{figure*}

\paragraph{Additional Experiments on Synthetic Dataset}
Since there are no ground-truth hyperedges in real-world datasets of learning informative hyperedges towards a specific outcome, we construct a synthetic dataset to verify if our model can recover the correct hyperedge structure under the MIMR objective. We use the precision, recall, and F1 score to measure the correctness of the learned hyperedges with respect to the ground truth. Although it is challenging to learn hyperedges when only supervised by the task label, we find that our model reaches high performances, with an average improvement of 28.3\% in terms of the F1 score, compared to the strongest baselines. Details about the synthetic experiments can be found in Appendix \ref{app:syn}.

\paragraph{Additional Experiments on Model Fit}
We further discuss the model's goodness of fit in Appendix \ref{app:fit}, with Mean Squared Error (MSE) as the evaluation metric. Our model outperforms the state-of-the-art model in 9 out of 11 datasets, with an average improvement of $11.9\%$.

%\rex{quantify: average of X improvement in ...} 

\subsection{Further Analysis}
% \subsubsection{Hyperedge Profiles}
In this subsection, we analyze the results of our model. We mainly use the ABCD dataset in the analysis since it is much larger than the ABIDE dataset.
\paragraph{Hyperedge Degree Distribution}
We plot the hyperedge degree distribution in Figure \ref{fig:degree_dist}. We find there are two distinct clusters in the figure. The first cluster is hyperedges with degree $\leq 5$. $1$-degree and $2$-degree hyperedges are special cases of our method: $1$-degree hyperedges are individual nodes, which imply the contribution of individual regions to the cognition. $2$-degree hyperedges reveal the importance of traditional pairwise connectivity. The other cluster concentrates around degree $25$, which implies the importance of relationships of multiple regions.

% To further understand the relationship between hyperedge degree and its significance, we next plot their correlation in Figure \ref{fig:degree_sig}. \jialin{can be removed}
% \begin{wrapfigure}{R}{0.49\textwidth}
%           \includegraphics[width=\linewidth]{figures/degree_sig.pdf}
%     \caption{Relationship between the significances and degrees of learned hyperedges.}
%     \label{fig:degree_sig}
% \end{wrapfigure}

\paragraph{Hyperedges with Higher Degree are More Significant}
Since CPM conducts a significance test (details can be found in Appendix~\ref{sec:app_cpm}) on pairwise edges and hyperedges internally based on a linear regression model, we can obtain a P-value for each hyperedge from the significance test. We define the significance of a hyperedge as $1-P_v \in [0, 1]$ where $P_v$ is the P-value of that hyperedge.

The relationship between hyperedge degree and its significance is shown in Figure ~\ref{fig:degree_sig}. In this figure, we find a strong positive correlation between a hyperedge's degree and its significance, which indicates that interactions of multiple brain regions play more important roles in cognition than pairwise or individual ones. It is also worth mentioning that there is a turning point around degree $5$, which corresponds to the valley around $5$ in Figure \ref{fig:degree_dist}.

\paragraph{High-order relationships are Better than Pairwise Ones}
To compare the significance in cognition between pairwise edges and learned hyperedges, we plot the number of remaining edges under different thresholds in Figure \ref{fig:cmp_sig_pair}. We find out that the learned hyperedges are much more significant than pairwise ones. Also note that even if we set the threshold to an extremely strict value ($1\times 10^{-8}$), there are still $60\%$ hyperedges considered significant. This evidence shows that our high-order relationships are much more significant than the traditional pairwise connectivity, which implies relationships involving multiple brain regions could be much more essential in cognition.

% \subsubsection{Region importance}

% \begin{table}
%     \centering
%     \begin{tabular}{cccccc}
%     \toprule
%         Condition & 1st & 2nd & 3rd & 4th & 5th\\
%         \midrule
%          % Rest & Insula\_R & Putamen\_L& Temporal\_Sup\_R& Amygdala\_R \\ 
%          SST & Fusiform\_L& Fusiform\_R& Frontal\_Inf\_Oper\_L& Acc\_pre\_L & Insula\_R \\
%          EN-back & Calcarine\_R& Cuneus\_R& Amygdala\_R& SupraMarginal\_L & Lingual\_L\\
%          MID & Fusiform\_R & Fusiform\_L& SupraMarginal\_R& Cuneus\_R & Insula\_R\\
%          \bottomrule
%     \end{tabular}
%     \caption{Top 5 important nodes under each condition.}
%     \label{tab:node_importance}
% \end{table}

% To figure out which regions play the most important roles under different conditions, we further define the importance of a node as how many hyperedges it belongs to. The Top 4 important nodes of each condition are shown in Table \ref{tab:node_importance}. A more complete ranking of regions can be found in Appendix \ref{app:region_importance}.

% \textbf{Observation 1.}

% \cam{
% \paragraph{Visualizations}

\paragraph{Hyperedge Case Study} We visualize the most significant hyperedge of the EN-back task in Figure \ref{fig:hyperedge1}. We observe a coordinated interaction of numerous brain regions, each fulfilling specific roles. Notably, some of these regions serve multi-functional purposes:
\begin{figure*}
    \centering
    \includegraphics[width=\linewidth]{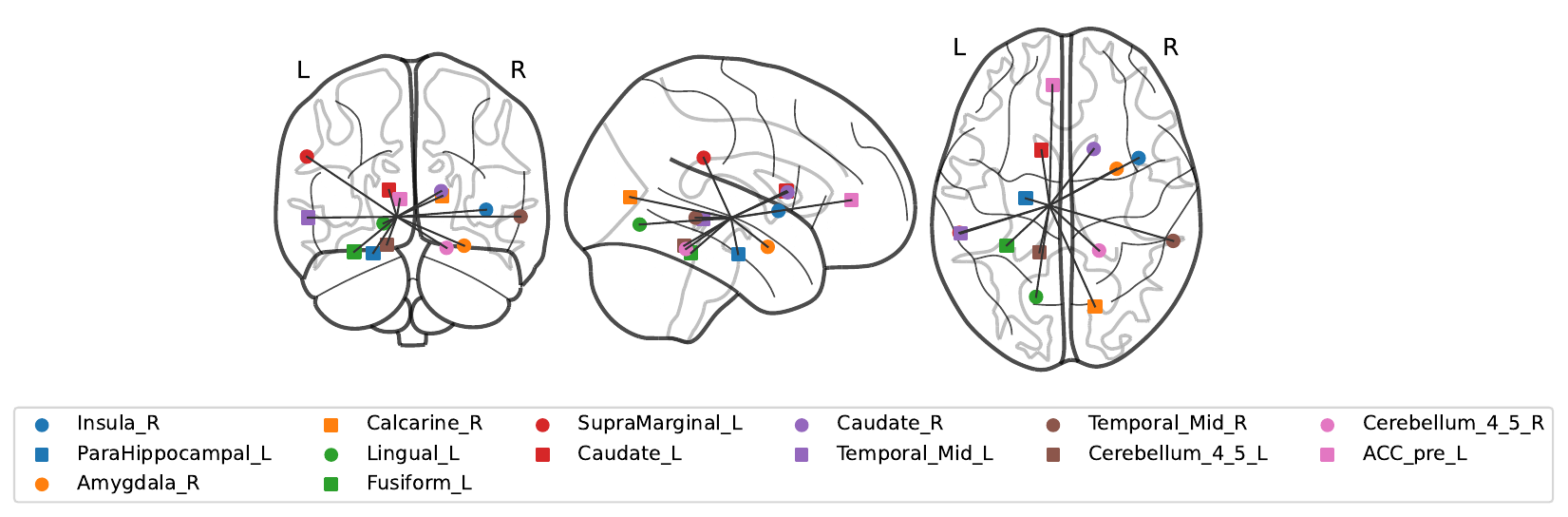}
    \caption{Visualization of the most significant hyperedge of the EN-back task.}
    \label{fig:hyperedge1}
\end{figure*}

\begin{itemize}
\item \textbf{Memory Processing} \textit{ParaHippocampal\_L}, \textit{Temporal Mid}: Essential for memory encoding and retrieval, these regions are integral to the EN-back task, facilitating the recall of previously viewed images.

\item \textbf{Emotional Processing} \textit{Amygdala\_R}: The amygdala is crucial for the processing of emotions, such as fear and pleasure. Since the EN-back task involves emotional stimuli, it is reasonable that the region is connected by the hyperedge.

\item \textbf{Visual Processing}: \textit{Calcarine\_R}, \textit{Lingual\_L}, \textit{Fusiform\_L}. These regions are responsible for visual perception and some of them are related to complex visual contents like symbols and human faces, which were presented during the fMRI task.

\item \textbf{Sensory} \textit{SupraMarginal\_L}: It is responsible for interpreting tactile sensors and perceiving limbs location. Its involvement is likely due to the requirement for participants to engage in specific physical actions, such as pressing buttons, during the task. \textit{Temporal Mid}: It functions in multi-modal sensory integration.

\item \textbf{Motor Control} \textit{Cerebellum}: It is primarily responsible for muscle control. \textit{Caudate}: It plays a crucial role in motor processes.  Its involvement is likely attributed to participants engaging in physical actions, like pressing buttons.

\item \textbf{Cognitive Control} \textit{ACC\_pre\_L}:  In the EN-back task, this region is likely crucial for maintaining focus, error detection and correction, conflict management in working memory, and modulating emotional responses to the task's demands.
\end{itemize}
% }

More visualizations about individual region importance can be found in Appendix \ref{app:viz}.

%% file: 005conclusion.tex
\section{Conclusion}
In this work, we proposed \mname{} for identifying maximally informative yet minimally redundant (MIMR) high-order relationships of brain regions. To effectively optimize our model, we further proposed a novel information bottleneck framework and derived its theory. Our method outperforms state-of-the-art models. The result analysis shows the effectiveness of our model. We expect such advancements could benefit clinical studies, providing insights into neurological disorders, and offering improved diagnostic tools in neurology and other related fields. %Limitations are discussed in Appendix \ref{app:limitations}.
\paragraph{Limitations}
\mname{} only considers static high-order relations. Given that fMRI tasks are dynamic, including temporal changes and interactions, it will be interesting to study the evolution of these high-order relationships.

Additionally, \mname{} does not offer a method for interpreting complex high-order relationships. This limitation is not specific to \mname{}, but is a common challenge in analyzing such relationships.We propose a hierarchical strategy that has the potential to interpret them to a certain extent, which is detailed in Appendix \ref{app:interpret}

\section*{Impact Statement}
This paper presents work whose goal is to advance the field of Machine Learning. There are many potential societal consequences of our work, none of which we feel must be specifically highlighted here.

%% file: appendix.tex
\section{Proof of the Upper Bound}
\label{app:proof1}
In this section, we prove the upper bound of $I(H; X)$ in \multiheadname{} in Equation \ref{equ:multi_drop}.{}

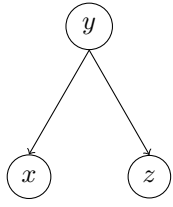
\begin{figure}[H]
    \centering
    \begin{tikzpicture}[
roundnode/.style={circle,draw},
]
%Nodes
\node[roundnode]   at (0,2)   (y)                              {$y$};
\node[roundnode] at (-0.8,0)     (x)     {$x$};
\node[roundnode]   at (0.8,0)  (z)     {$z$};
%Lines
\draw[->] (y.south) -- (x.north);
\draw[->] (y.south) -- (z.north);
\end{tikzpicture}

    \caption{The graphical model of random variables $X$, $Y$ and $Z$.}
    \label{fig:graphical}
\end{figure}

\begin{lemma}
    \label{lemma:1}
    Given random variables $X$, $Y$ and $Z$. Their relationships are described in the graphical model illustrated in Figure \ref{fig:graphical}. We have
\begin{equation}
    I(X; Y | Z) \leq I(X; Y)
\end{equation}
\end{lemma}
\begin{proof}

\begin{equation}
    \begin{aligned}
I(X; &Y | Z) - I(X; Y) \\
&= \int p(x, y, z) \log \frac{p(z)p(x, y, z)}{p(x, z)p(y, z)} \mathrm{d}x \mathrm{d}z \mathrm{d}y \\ &- \int p(x, y) \log \frac{p(x, y)}{p(x)p(y)} \mathrm{d}x \mathrm{d}y\\
&=\int p(x, y, z) \log \frac{p(z)p(x, y, z)}{p(x, z)p(y, z)} \mathrm{d}x \mathrm{d}z \mathrm{d}y \\ &- \int p(x, y, z) \log \frac{p(x, y)}{p(x)p(y)} \mathrm{d}x \mathrm{d}y \mathrm{d}z\\
&=\int p(x, y, z) \log \frac{p(z)p(x,y,z)p(x)p(y)}{p(x,z)p(y,z)p(x,y)}\mathrm{d}x \mathrm{d}z \mathrm{d}y\\
&=\int p(x, y, z) \log \frac{p(z)p(x,z|y)p(y)p(x)p(y)}{p(x,z)p(y,z)p(x,y)}\mathrm{d}x \mathrm{d}z \mathrm{d}y\\
&=\int p(x, y, z) \log \frac{p(z)p(x|y)p(z|y)p(y)p(x)p(y)}{p(x,z)p(y,z)p(x,y)}\mathrm{d}x \mathrm{d}z \mathrm{d}y\\
&=\int p(x, y, z) \log \frac{p(x)p(z)}{p(x,z)}\mathrm{d}x \mathrm{d}z \mathrm{d}y\\
&=\int p(x, z) \log \frac{p(x)p(z)}{p(x,z)}\mathrm{d}x \mathrm{d}z\\
&=-I(X; Z) \leq 0
    \end{aligned}
\end{equation}
\\which finishes the proof.

\end{proof}
\begin{corollary}
\label{cor:ib}
    Given the same graphical model \ref{fig:graphical}, we have
    \begin{equation}
        I(X, Z; Y) \leq I(X; Y) + I(Z; Y)    
    \end{equation}

\end{corollary}
\begin{proof}
Using the chain rule of mutual information, we obtain
\begin{equation}
    I(X,Z;Y) = I(X;Y|Z) + I(Z; Y)
\end{equation}
According to Lemma~\ref{lemma:1}, we have $I(X; Y|Z) \leq I(X; Y)$, which finishes the proof.

\end{proof}

\begin{theorem}
    For random variables in Equation~\ref{equ:multi_drop}, we have
    \begin{equation}
        I(H;X) \leq \sum_{k=1}^K I(\boldsymbol{h}^k; X)
    \end{equation}
\end{theorem}
\begin{figure}[H]
    \centering
        \begin{tikzpicture}[
roundnode/.style={circle,draw, minimum size=0.8cm},
]
%Nodes
\node[roundnode]   at (0,2)   (x)                              {$x$};
\node[roundnode] at (0, 0)     (hk)     {$\boldsymbol{h}^k$};
\node[draw, minimum size=2cm] at (0,0) (r) {};
\node[] at (0.7,-0.7) (t) {$K$};

%Lines
\draw[->] (x.south) -- (hk.north);
\end{tikzpicture}
    \caption{The graphical model of random variables $\boldsymbol{h}^k$ and $X$}
    \label{fig:graphical2}
\end{figure}
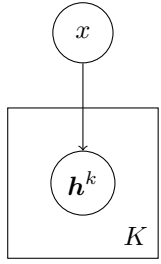

\begin{proof}
    According to the definitions of $X$ and $H$, which are described in Section \ref{sec:approach}, we can draw a graphical model of them in Figure~\ref{fig:graphical2}. Define a new random variable $\boldsymbol{h}^{k_1:k_2} = [\boldsymbol{h}^{k_1}, \boldsymbol{h}^{k_1+1}, \cdots, \boldsymbol{h}^{k_2-1}, \boldsymbol{h}^{k_2}]$, which is a concatenation from $\boldsymbol{h}^{k_1}$ to $\boldsymbol{h}^{k_2}$. According to Corollary~\ref{cor:ib} we have
    \begin{equation}
    \begin{aligned}
        I(H; X) &\leq I(\boldsymbol{h}^1, X) + I(\boldsymbol{h}^{2:K}, X)\\
        &\leq I(\boldsymbol{h}^1, X) + I(\boldsymbol{h}^2, X) + I(\boldsymbol{h}^{3:K}, X)\\
        &\leq I(\boldsymbol{h}^1, X) + I(\boldsymbol{h}^2, X) + I(\boldsymbol{h}^3, X) + \cdots \\
        &\leq \sum_{k=1}^K I(\boldsymbol{h}^k; X)
    \end{aligned}
    \end{equation}
\end{proof}

\begin{theorem} (proposition 1)
    \begin{equation}
            \begin{aligned}
        I(H; X)
        \leq \sum_{k=1}^K \sum_{i=1}^N
        I(\boldsymbol{h}^k_i; X_i)
        = \sum_{k=1}^K\sum_{i=1}^N \mathbb{H}[X_i] (1 - p^k_{\theta, i})
    \end{aligned}
    \end{equation}
\end{theorem}

\begin{proof}
Given \textbf{Theorem 2}. It suffices to prove
\begin{equation}
    I(\boldsymbol{h}^k; X) \leq \sum_{i=1}^N I(\boldsymbol{h}^k_i; X_i) = \sum_{i=1}^N \mathbb{H}[X_i] (1 - p^k_{\theta, i}),\quad \forall 1\leq k \leq K
\end{equation}
And this is exactly the conclusion in \cite{kim2021drop} if we consider $\boldsymbol{h}^k$ and $X$ as $X$ and $Z$ respectively in their paper.
\end{proof}

\section{Connectome-Based Predictive Modeling}
\label{sec:app_cpm}

\citet{shen2017using,finn2015functional} have shown tremendous promise in recent years in detecting imaging biomarkers by CPM (connectome-based predictive modeling) \citep{rosenberg2015predicting, dubois2018resting, rosenberg2020functional, rosenberg2016neuromarker}. Such a model, based chiefly on functional MRI data, can measure the significance of the input edge weights, which is revealed by a correlation coefficient that reflects the correlation between the edge weights and the neurological outcomes. One could expect a large correlation coefficient to indicate the high quality of edge weights. We utilize the CPM as an evaluation model to evaluate the quality of our learned hyperedges.
%The Connectome-based Predictive Modeling (CPM) pipeline is a methodological framework designed to predict neurological outcomes, such as cognitive scores, by analyzing brain connectivity.
Here is a pipeline overview of the CPM process:

\begin{enumerate}[topsep=0pt,itemsep=-1ex,partopsep=1ex,parsep=1ex]
    \item \textbf{Connectivity Calculation}: For each subject, compute the Pearson correlation coefficients for each possible pair of brain regions. This is based on the fMRI series of those regions.
    \item \textbf{Edge Significance}: Calculate the correlation between each brain connectivity edge and the outcome of interest (e.g., cognition scores) across all subjects. The correlation of an edge indicates its significance.
    \item \textbf{Edge Selection}: Identify significant connectivity edges. These are the edges where the correlation values are greater than a predetermined significance threshold.
    \item \textbf{Weight Summation}: For each subject, sum the weights of the significant edges identified in the previous step to derive a single summary score (scalar).
    \item \textbf{Model Fitting}: Fit a linear model that predicts the neurological outcomes based on the summed weights, where each subject is a sample.
    \item \textbf{Model Evaluation}: Across all subjects, calculate the correlation of predicted values and the neurological outcomes. Note that this correlation coefficient is equivalent to the one between the summed weights and the outcomes, and is exactly the metric $r$ we use to evaluate our hyperedges in Equation \ref{equ:metric}.
\end{enumerate}
Since positive edges and negative edges will cancel out with each other when being summed, we adopt the combining strategy in \cite{boyle2023connectome}.

\paragraph{CPM Measures the Quality of Edge Weights}
According to step 6, the evaluation of the predictive model could be measured by the correlation between predicted and ground-truth outcomes \cite{shen2017using}. Since CPM is a linear model that predicts the outcome based on the sum of significant edge weights, the correlation is equal to the correlation between the sum and ground-truth outcomes (which is exactly the $r$ in Equation \ref{equ:metric}). Hence, one can expect a larger correlation if the edge weights are more correlated (and thus are more predictive).

\paragraph{Significance of Edges in CPM}
In step 2, CPM obtains a correlation coefficient $r^k$ for each edge weight $w^k$ and the cognition score $Y$ across all subjects. Consider a classical hypothesis test $H_0: r^k = 0, H_1: r^k \neq 0$. Assume $w^k$ and $Y$ are drawn from independent normal distribution (corresponds to $H_0$), the probability density function of correlation coefficient $r^k$ is
$$
f(r^k)=\frac{\left(1-{r^k}^2\right)^{n / 2-2}}{\mathrm{~B}\left(\frac{1}{2}, \frac{n}{2}-1\right)},
$$
where $n$ is the number of samples and $B$ is the beta function. Based on the distribution, we obtain a $P$-value for the $k$-th edge, which is used to measure the significance of the edge.

\section{Computational Complexity}
\label{app:complexity}
Suppose we have $N$ regions as $N$ nodes and $K$ hyperedges. For the \textsc{Constructor}, the unignorable computation is from the mask operation. The computational complexity of this step for each hyperedge is $O(Nd)$ since it does a pairwise multiplication operation of a matrix of size $N \times d$. Given that there are $K$ hyperedges, the total complexity is $O(NKd)$. For the \textsc{Weighter}, the computation is from the dim reduction operation and the linear head. The dim reduction operation is an MLP. In this work, the hidden dimensions are a fraction of the original feature dimension. Therefore, the complexity of the dim-reduction MLP is $O(d^2)$. The linear head only contributes $O(d)$, which is neglectable. 
As a result, the computational complexity of the whole model is $O(NKd + d^2) = O(N^2K)$ since the feature dimension is equal to the number of regions (See Appendix \ref{app:dataset} for features we used). This complexity is just at the same scale as that of MLPs even though we are addressing a more challenging task: identifying high-order relationships in an exponential space.

\section{Dataset details}
\label{app:dataset}
\subsection{Preprocessing of ABCD Dataset}
% The functional MRI (fMRI) data is collected from children when they were resting capturing intrinsic brain activity (Rest), when they were performing the emotional n-back task (EN-back), the Stop Signal task (SST), and the Monetary Incentive Delay (MID) task.
We use the preprocessed ABIDE dataset from the official website. We preprocessed the restricted ABCD dataset ourselves. Below is the preprocessing procedure of the ABCD dataset.

\paragraph{raw data to voxel-level fMRI time series}
The fMRI data is processed using BioImage Suite \citep{joshi2011unified}. First, we performed motion correction and slice-time correction using SPM5; and via BioImage Suite, the data were registered to a standardized $3mm \times 3mm \times 3mm$ common space, where we generated masks representing white matter, gray matter, and cerebrospinal fluid (CSF) and computed the mean time courses for both white matter and CSF. We orthogonalized each gray matter time course with respect to the mean time courses of both white matter and CSF, and we orthogonalized each gray matter time course to the six motion-related signals via SPM5. We then applied a bandpass Butterworth filter with a frequency range of $0.02$Hz to $0.1$Hz to the orthogonalized time courses. We used a Gaussian kernel with a full-width at half-maximum (FWHM) of $6$mm to enhance spatial coherence and spatial smoothing. Lastly, we removed the linear trend from all signals in accordance with the methodology detailed in \cite{shenGroupwiseWholebrainParcellation2013}. We deleted scans with more than 0.10 mm mean frame-to-frame displacement. Additional details about the standard preprocessing procedures, such as slice time and motion correction, and registration to the MNI template can be found in \cite{greene2018task} and \cite{horien2019individual}.
\paragraph{voxel-level fMRI time series to region features}
The fMRI time series data of a human subject is represented in four dimensions (3 spatial dimensions + 1 temporal dimension), which can be imagined as a temporal sequence of 3D images. First, brain images are parceled into regions (or nodes) using the AAL3v1 atlas \citep{rolls2020automated}. Following previous works \citep{kan2022brain,li2021braingnn,thomas2022self}, each region's time series is obtained by averaging all voxels in that region. Consistent with previous connectivity-based methods \citep{li2021braingnn,kan2022brain,ktena2018metric,said2023neurograph}, for each region, we use its Pearson correlation coefficients to all regions as its features. We randomly split the data into train, validation, and test sets in a stratified fashion. The split ratio is 8:1:1.

\subsection{Dataset Statistics}
The statistics of the number of instances and the time series length are summarized in Table \ref{tab:stat}.
\begin{table*}
    \centering
     \caption{Statistics of datasets we use. ABIDE dataset only contains resting-state fMRI data. ABCD contains resting-state and task-based data. We use ABCD of $2$ timepoints. For example, \textit{Rest 1} means resting-state fMRI data from timepoint $1$.}
    \label{tab:stat}
    \resizebox{\textwidth}{!}{
    \begin{tabular}{ccccccccccccc}
    \toprule
          % Time Point&\multicolumn{4}{c}{1}&  &\multicolumn{4}{c}{2}\\
          % \cmidrule{2-5}\cmidrule{7-10}
          \multirow{2}{*}{Dataset} & \multicolumn{3}{c}{ABIDE} & & \multicolumn{8}{c}{ABCD} \\
          \cmidrule(lr){2-4} \cmidrule(lr){6-13}
           & FIQ & VIQ & PIQ & &Rest 1&  SST 1&  EN-back 1&  MID 1& Rest 2&  SST 2&  EN-back 2& MID 2\\
          \midrule
          \#(instances) & 1035  & 1035  & 1035&  & 1676 &  1673&  1678&  1678& 1949&  1053&  1044& 1062\\
          length of time series & 196 & 196 & 196& & 375 & 437 & 362 & 403 & 375 & 437 & 362 & 403 \\
          \bottomrule
    \end{tabular}
    }
\end{table*}

\section{Training Details}
\label{app:training}
Due to the data scarcity, we found training on single sub-datasets of the ABCD dataset leads to severe overfitting. To mitigate this, we train our model as well as all the baselines on all datasets together and report the results individually. Note that on the ABIDE dataset, we train our model under the three targets separately since we don't encounter such an issue on the ABIDE dataset.

\paragraph{Hardware}
We train our model on a machine with an Intel Xeon Gold 6326 CPU and RTX A5000 GPUs.

\paragraph{Software}
See Table \ref{tab:software} for the software we used and the versions.
\begin{table}[H]
    \centering
    \caption{Software versions.}
    \label{tab:software}
    \begin{tabular}{cc}
    \toprule
    software & version\\
    \midrule
    python & 3.8.13\\
       pytorch  & 1.11.0 \\
       cudatoolkit & 11.3 \\
       numpy  & 1.23.3 \\
       ai2-tango & 1.2.0\\
       nibabel & 4.0.2\\
       \bottomrule
    \end{tabular}
\end{table}

\paragraph{Hyperparameter Choices}
The hyperparameters selection is shown in Table \ref{tab:hyper}. Some crucial hyperparameters ablation experiments can be found in Appendix \ref{app:ablation}.

\begin{table}
    \centering
    \caption{Hyperparameter choices.}
    \label{tab:hyper}
    \begin{tabularx}{\columnwidth{}}{cXc}
    \toprule
         notation&  meaning& value\\
         \midrule
         $lr$&  learning rate& $1\times 10^{-3}$\\
         $K$&  number of hyperedges& $32$\\
         $\beta$& trade-off coefficients information bottleneck& $0.2$\\
         $[h_1, h_2, h_3]$&  hidden sizes of the dim reduction MLP& $[32,8,1]$\\
         $B$&  batch size& $64$\\
         \bottomrule
    \end{tabularx}
\end{table}

\section{Why is \textsc{HyBRiD} not the Best on the Rest 1 Dataset?}
\label{app:why_fail}
At present, we do not fully understand the reasons for the phenomena observed. The hypotheses outlined below are based on our preliminary observations and analyses.

\paragraph{Hypothesis 1: Sparsity in Resting-State Connectivity}
We propose that resting-state connectivity is sparser and more diffuse compared to task-based connectivity. This assertion is refined from our initial claim that the connectivity relevant to the phenotypic outcomes during the resting state is notably sparser than during the task-based state.

Upon examining traditional pairwise connections, we summarized the number of significant connections most related to the phenotypic outcome in Table~\ref{tab:num_sig_fail}. Our analysis indicates that Rest 1 exhibits the fewest phenotype-related connections, suggesting that these connections are indeed sparse. We did not evaluate high-order relationships due to the lack of ground truth data for such connections. Further investigation is required to validate this hypothesis. 

\begin{table}[H]
    \centering
    \begin{tabular}{c|cccccccc}
    \toprule
	& SST $1$ 	&EN-back $1$ 	&MID $1$ 	&Rest $1$ 	&SST $2$ 	&EN-back $2$ 	&MID $2$ 	&Rest $2$ \\
    \midrule
\# (connections) 	&$3164$ 	&$3414$ 	&$2989$ 	&$2036$ 	&$3112$ 	&$3014$ 	&$2692$ 	&$3896$\\
\bottomrule
    \end{tabular}
    \caption{The number of significant pairwise edges selected by CPM in different tasks and timepoints.}
    \label{tab:num_sig_fail}
\end{table}

\paragraph{Hypothesis 2: Comparative Quality of Data Across Timepoints}
We hypothesize that the data quality at timepoint 1 is inferior to that at timepoint 2. The performance metrics, as shown in the Table \ref{tab:train_test}, are better at timepoint 2 than at timepoint 1. Moreover, timepoint 1 exhibits more severe overfitting compared to timepoint 2. This discrepancy may be attributed to the higher relevance of the information (pertaining to the outcome) and lower noise levels at timepoint 2. This hypothesis also requires further verification through detailed investigations and experiments.
\begin{table}[H]
    \centering
    \begin{tabular}{c|cccccccc}
    \toprule
	&SST $1$ 	&EN-back $1$ 	&MID $1$ 	&Rest $1$ 	&SST $2$ 	&EN-back $2$ 	&MID $2$ 	&Rest $2$ \\
 \midrule
    train 	&$0.988$ 	&$0.984$ 	&$0.984$ 	&$0.984$ 	&$0.982$ 	&$0.977$ 	&$0.978$ 	&$0.976$ \\
test 	&$0.361$ 	&$0.348$ 	&$0.386$ 	&$0.223$ 	&$0.738$ 	&$0.714$ 	&$0.816$ 	&$0.730$\\
\bottomrule
    \end{tabular}
    \caption{Performances on training and test dataset in different tasks and timepoints.}
    \label{tab:train_test}
\end{table}

\paragraph{Conclusion}
Brain activities during resting states are not driven by external tasks, leading to more diffuse and less predictable patterns of activation. The low data quality of ABCD timepoint 1 (compared to timepoint 2) even intensifies this issue. Therefore, high-order relations might not be a good inductive bias on the Rest 1 dataset since the connections might be much more sparse and involve fewer nodes.

% Our model reaches the second-best on the Rest 1 dataset. To find out the reason, we inspect the performance on both the training, validation and test datasets. We summarize the performances of our model and the SOTA one in the table below
% \begin{table}[H]
%     \centering
%     \caption{Our model's performance compared to the SOTA one on the Rest 1 dataset.}
%     \label{tab:why_fail}
%     \begin{tabular}{cccc}
%     \toprule
%     model & train & validation & test \\
%     \midrule
%     BrainNetTF & $0.967$ & $0.484$ & $0.334$ \\
%     \mname{} (Ours)   & $0.984$ & $0.491$ & $0.223$ \\
%     \bottomrule
%     \end{tabular}
% \end{table}
% We find our model outperforms BrainNetTF on both training and validation datasets but fails on test datasets. Therefore, given that Rest 1 is the most noisy dataset, we argue that the imperfection is because

% 1) Our model overfits the training dataset more than BrainNetTF.

% 2) The validation score is higher, indicating the reason is the discrepancy between the training dataset and the test dataset.

\section{Evaluation on Synthetic Dataset}
\label{app:syn}
\paragraph{Dataset Synthesis}

\begin{table*}[htbp]
    \centering
    \caption{Precision, recall and F1 score on synthetic dataset under different $K$, where $K$ is the number of hyperedges. The performance of the Random method increases as $K$ increases because the Hungarian algorithm is likely to provide better matches when there are more candidates.}
    \label{tab:syn}
    \resizebox{\linewidth}{!}{
    \begin{tabular}{cccccccccccccccc}
    \toprule
         \multirow{2}{*}{Metric} & \multicolumn{3}{c}{$K=1$} & & \multicolumn{3}{c}{$K=5$} && \multicolumn{3}{c}{$K=10$} && \multicolumn{3}{c}{$K=30$} \\
         \cmidrule(lr){2-4}
         \cmidrule(lr){6-8}
         \cmidrule(lr){10-12}
         \cmidrule(lr){14-16}
        
         & P & R & F1 && P & R & F1&& P & R & F1& & P & R & F1  \\
         \midrule
         Random & $0.158$ & $0.094$ & $0.110$ && $0.182$ & $0.170$ & $0.171$ && $0.224$ & $0.218$ & $0.218$& & $\underline{0.288}$ & $0.250$ & $0.266$ \\
        $k$NN  & $0.533$ & $\mathbf{1.000}$ & $0.696$ && $0.247$  & $0.241$ & $0.243$ & &$0.144$ & $0.156$ & $0.148$ && $0.139$ & $0.128$ & $0.133$ \\
        $l_1$ hypergraph & $0.160$ & $\mathbf{1.000}$ & $0.276$  && $0.180$ & $0.016$ & $0.028$ && $0.135$ & $0.016$ & $0.028$ && $0.256$ & $0.026$ & $0.047$ \\
        $l_2$ hypergraph &  $\underline{0.987}$ & $0.825$  & $\underline{0.897}$ &&  $\underline{0.412}$ & $\underline{0.631}$ & $\underline{0.494}$ && $\underline{0.276}$ & $\underline{0.499}$ & $\underline{0.351}$ && $0.233$ & $\underline{0.470}$ & $0.310$ \\
        HGCN & $0.625$ & $\mathbf{1.000}$ & $0.769$  && $0.196$ & $0.477$ & $0.278$ && $0.188$ & $0.379$  & $0.251$ && $0.277$ & $0.364$ & $\underline{0.314}$  \\
        \shline
         \mname{} (Ours) & $\mathbf{1.000}$ & $\mathbf{1.000}$ & $\mathbf{1.000}$ && $\mathbf{0.595}$  & $\mathbf{0.680}$ & $\mathbf{0.631}$ && $\mathbf{0.504}$&$\mathbf{0.604}$ & $\mathbf{0.549}$ && $\mathbf{0.428}$ & $\mathbf{0.393}$ & $\mathbf{0.401}$ \\
         \bottomrule    
    \end{tabular}
    }
\end{table*}

Our synthetic dataset is constructed as follows:

1) \textit{Structure Generation}: For each hyperedge, we randomly sample the hyperedge degree $d$ from a discrete uniform distribution between $2$ and $d_{\text{max}}$. After that, we randomly sample $d$ nodes as members of this hyperedge. 

2) \textit{Feature generation}: For each hyperedge, we randomly sample a scalar value $v$ from the uniform distribution $U(0,1)$. We then randomly sample the features from $U(0, 2v)$ for each node in this hyperedge.

3) \textit{Label Generation}: For each hyperedge, we calculate the maximum of its node features as the summary of the hyperedge. We sum the summaries of all hyperedges to get a single value $Y$ as the label of the hypergraph. 

Note that consistent with our settings of real-world datasets, the structure is shared across all hypergraphs, while node features and hyperedge weights are different.

\paragraph{Metrics} We use the macro-precision, recall and F1 score to measure the correctness of the learned hyperedges with respect to the ground-truth ones. Note that when learning on the synthetic dataset, the order of hyperedge may differ from the ground truths. Therefore, the Hungarian algorithm is employed to match the learned hyperedges with the ground truth. The precision, recall and F1 score are calculated after matching.

\paragraph{Baselines}
We use the hyperedge construction methods in Section \ref{sec:experiments}, i.e. $k$NN\cite{huang2009video}, $l_1$ hypergraph\cite{wang2015visual} and $l_2$ hypergraph\cite{jin2019robust} as our baselines. 
Besides, we implement a hypergraph structure learning model. This is a hypergraph convolutional neural network (HGCN) \citep{bai2021hypergraph}, where both the parameters of layers and the hypergraph structure (the incidence matrix) are learnable.

\paragraph{Implementaton Details}
The maximum degree $d_{\text{max}}$ is set to $34$, which is the maximum degree of hyperedges of the ABCD dataset according to the analysis in Section \ref{sec:experiments}. The number of nodes is set to $164$, which is the number of regions of AAL3v1 atlas, used in the ABCD experiment. We conduct the experiments under different numbers of hyperedges (i.e. $K = 1,5,10,30$).

\paragraph{Results}
We report the results under different numbers of hyperedges. Although this is a hard task, our model consistently outperforms all the baselines significantly, with an average improvement of 28.3\% in terms of the F1 score. This demonstrates that our model can learn the hyperedges well under the MIMR objective, without the direct supervision of hyperedge ground truths. This also inspires us to use automated model evaluation techniques \cite{wang2023toward, peng2023came, peng2024energy} on real-world data in the future where labelled high-order relationships are not readily accessible.

\section{Model Fit Performance}
We report the goodness of fit of our model and the state-of-the-art baseline in Table~\ref{tab:fit_abide}, with the Mean Square Error (MSE) as the metric.
\label{app:fit}
\begin{table}[H]
    \centering
    \begin{tabular}{c|ccc}
    \toprule
        Model & FIQ & VIQ & PIQ \\
    \midrule
%         BrainNetTF & $0.097$  &	$0.089$ & $0.055$ \\
% \mname{} (ours) & $0.057$ &	$0.071$ &	$0.046$ \\
        BrainNetTF & $5.917$  &	$5.429$ & $3.355$ \\
\mname{} (ours) & $\mathbf{3.477}$ &	$\mathbf{4.331}$ &	$\mathbf{2.806}$ \\
\bottomrule
    \end{tabular}
    \caption{Performance of the model and baselines (MSE) in fitting the target on the ABIDE dataset. }
    \label{tab:fit_abide}
\end{table}
% \todo{bolden and multiply the numbers by 61}

\begin{table}[H]
    \centering
    \begin{tabular}{c|cccccccc}
    \toprule
Model &	SST 1 &	EN-back 1 &	MID 1 &	Rest 1 &	SST 2 &	EN-back 2 &	MID 2 &	Rest 2 \\
\midrule
% BrainNetTF &	0.0189 &	0.0225 &	0.0188& 	0.0204& 	0.0101& 	0.0111& 	0.0118& 	0.0123 \\
% \mname{} (ours) & 	0.0169& 	0.0221& 	0.0209& 	0.0241& 	0.0099& 	0.0105& 	0.0089& 	0.0076 \\

BrainNetTF & 1.153 & 1.373 & $\mathbf{1.147}$ & $\mathbf{1.244}$ & 0.616 & 0.677 & 0.720 & 0.750\\
\mname{} (ours) &  $\mathbf{1.031}$& 	 $\mathbf{1.348}$&  $1.275$& 	 $1.470$& 	 $\mathbf{0.604}$& 	 $\mathbf{0.641}$&  $\mathbf{0.543}$& 	 $\mathbf{0.464}$ \\
\bottomrule
    \end{tabular}
    \caption{Performance of the model and baselines (MSE) in fitting the target on the ABCD dataset.}
    \label{tab:fit_abcd}
\end{table}
% \todo{bolden and multiply the numbers by 61}

\section{More Ablation Studies}
\label{app:ablation}

\subsection{Choices of Number of Hyperedges $K$}
\label{app:k_choice}
As explained in Section \ref{sec:approach}, we use $K$ heads for $K$ hyperedges. We study the correlation between the $r$ value and the number of hyperedges on three datasets:
\begin{figure}[H]
    \begin{subfigure}[t]{0.48\linewidth}
            \includegraphics[width=\textwidth]{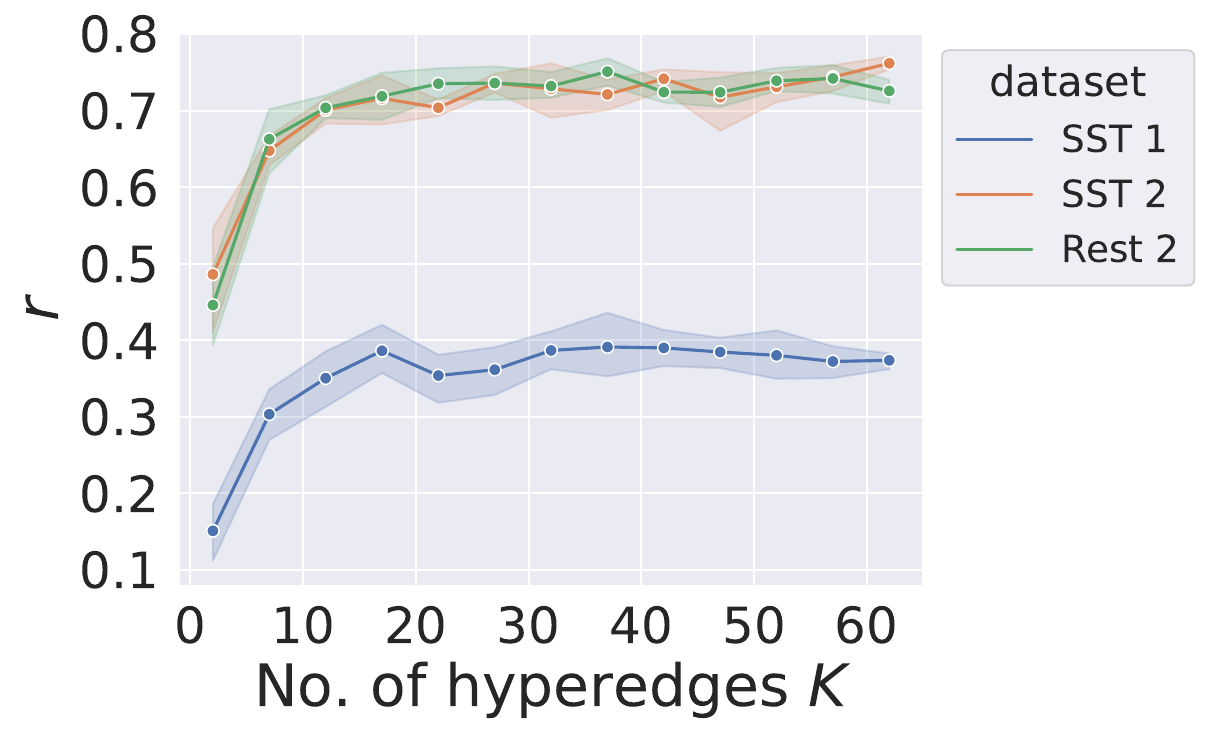}
    \caption{The performance with increasing the number of hyperedges through $2$, $7$, $12$, $17$, $22$, $27$, $32$, $37$, $52$, $57$, $62$ on three datasets.}
    \label{fig:abl_hyper}
    \end{subfigure}
    \hfill
    \begin{subfigure}[t]{0.48\linewidth}
    \includegraphics[width=\textwidth]{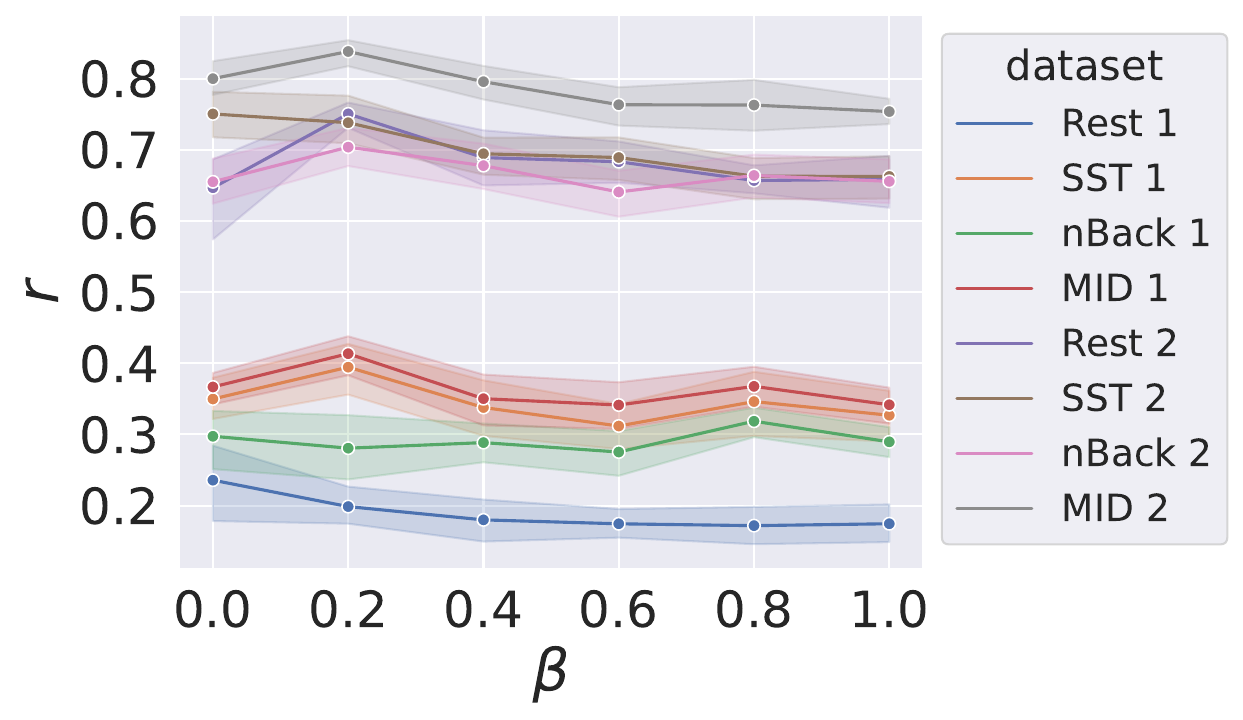}
    \caption{Performance with increasing $\beta$ from $0.1$ to $1.0$ with step $0.1$ on all datasets.}
    \label{fig:abl_beta}
    \end{subfigure}
    \caption{Studies on the choice of two key hyperparameters, $\beta$ and $K$, in our model.}

\end{figure}

From Figure \ref{fig:abl_hyper}, we find that the overall performance increases dramatically before $K=17$, but becomes stable and close to saturation after $K=32$. To improve the efficiency while ensuring the performance, we choose $K=32$.

\subsection{Choices of the Trade-off Coefficient $\beta$}
\label{app:beta}

In our optimization objective \ref{equ:ib2}, $\beta$ acts as a trade-off parameter, which is a non-negative scalar that determines the weight given to the second term relative to the first. To study its fluence to the performance, we plot the model performances on all datasets under different $\beta$ in Figure \ref{fig:abl_beta}. We can see performances on $3$ datasets (Rest 1, SST 2) consistently decrease when $\beta$ increases. However, on other $5$ datasets (SST 1, MID 1, Rest 2, nBack 2, MID 2), we can observe a peak at $\beta=0.2$. Accordingly, we adopt $\beta=0.2$.

% \section{Region Importance}
% To better understand each brain region's role in cognition under different conditions. We define the region importance as how many times it gets involved (participates) in a hyperedge. The more it participates, the more important it should be considered. We rank these regions by their importance in Figure \ref{fig:region_importance}. 

% \label{app:region_importance}
% \begin{figure}
% \begin{subfigure}[b]{0.48\textwidth}
%     \includegraphics[width=\textwidth]{figures/most_Rest_participants.pdf}
%     \caption{Node importance ranking of Rest.}
% \end{subfigure}
% \begin{subfigure}[b]{0.48\textwidth}
%     \includegraphics[width=\textwidth]{figures/most_SST_participants.pdf}
%     \caption{Node importance ranking of SST.}
% \end{subfigure}
% \begin{subfigure}[b]{0.48\textwidth}
%     \includegraphics[width=\textwidth]{figures/most_nBack_participants.pdf}
%     \caption{Node importance ranking of EN-back.}
% \end{subfigure}\begin{subfigure}[b]{0.48\textwidth}
%     \includegraphics[width=\textwidth]{figures/most_MID_participants.pdf}
%     \caption{Node importance ranking of MID.}
% \end{subfigure}
% \caption{Node importance ranking under different conditions.}
% \label{fig:region_importance}
% \end{figure}

\section{Runtime Comparison}
\label{app:runtime}
Figure \ref{fig:runtime} summarizes the per-batch training time of all deep learning models. We find that \mname{} is the most efficient one, with $87\%$ faster than the second one (BrainNetTF) and at least $1255\%$ faster than the GNN-based ones.

\begin{figure}[H]
    \centering
    \includegraphics[width=0.5\linewidth]{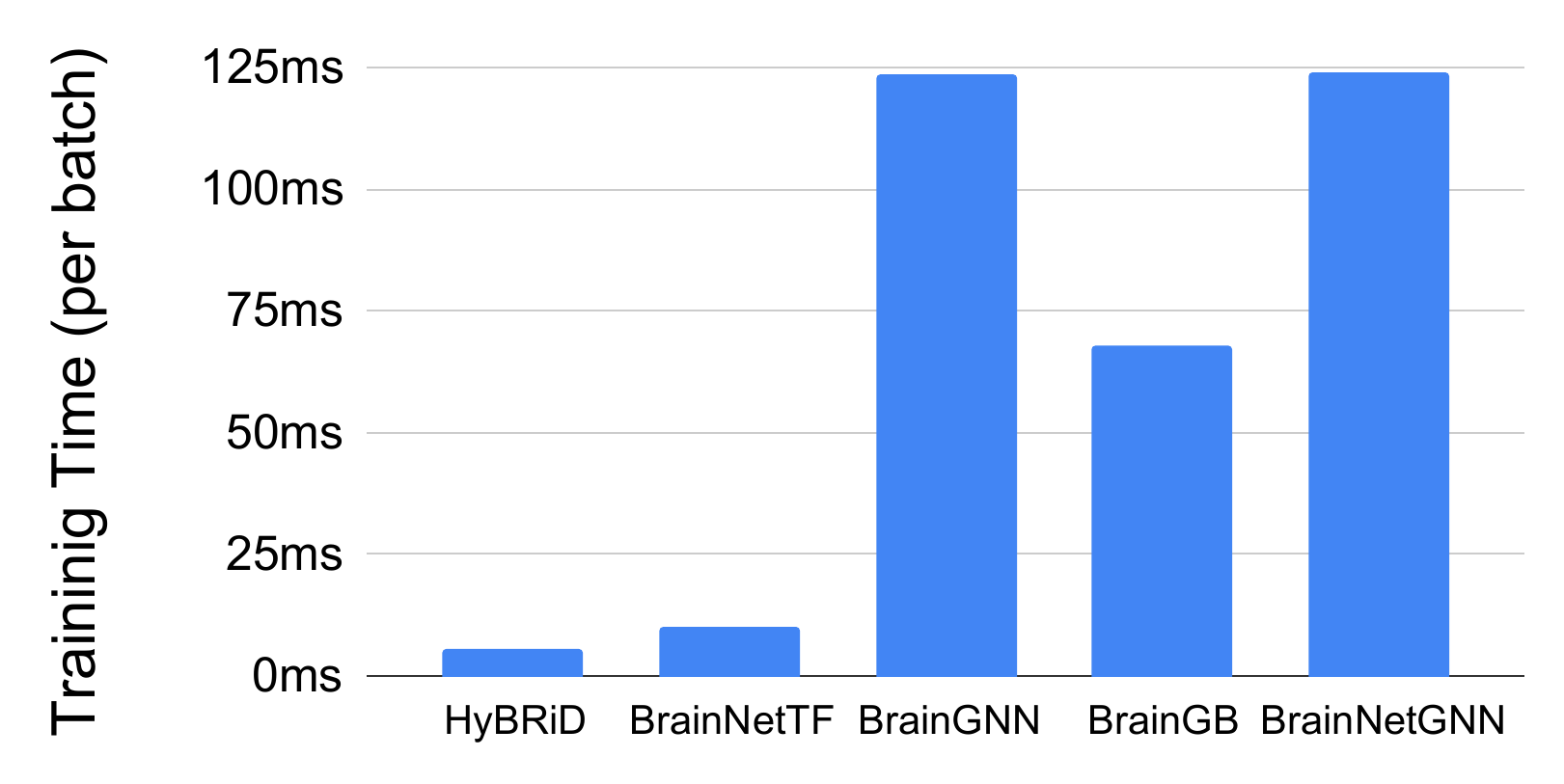}
    \caption{Training time per batch of all deep learning models.}
    \label{fig:runtime}
\end{figure}

\section{More Visualizations}
\label{app:viz}
\paragraph{Task-Based Brain Region Importance}
\begin{figure}
    \centering
    \includegraphics[width=0.49\textwidth]{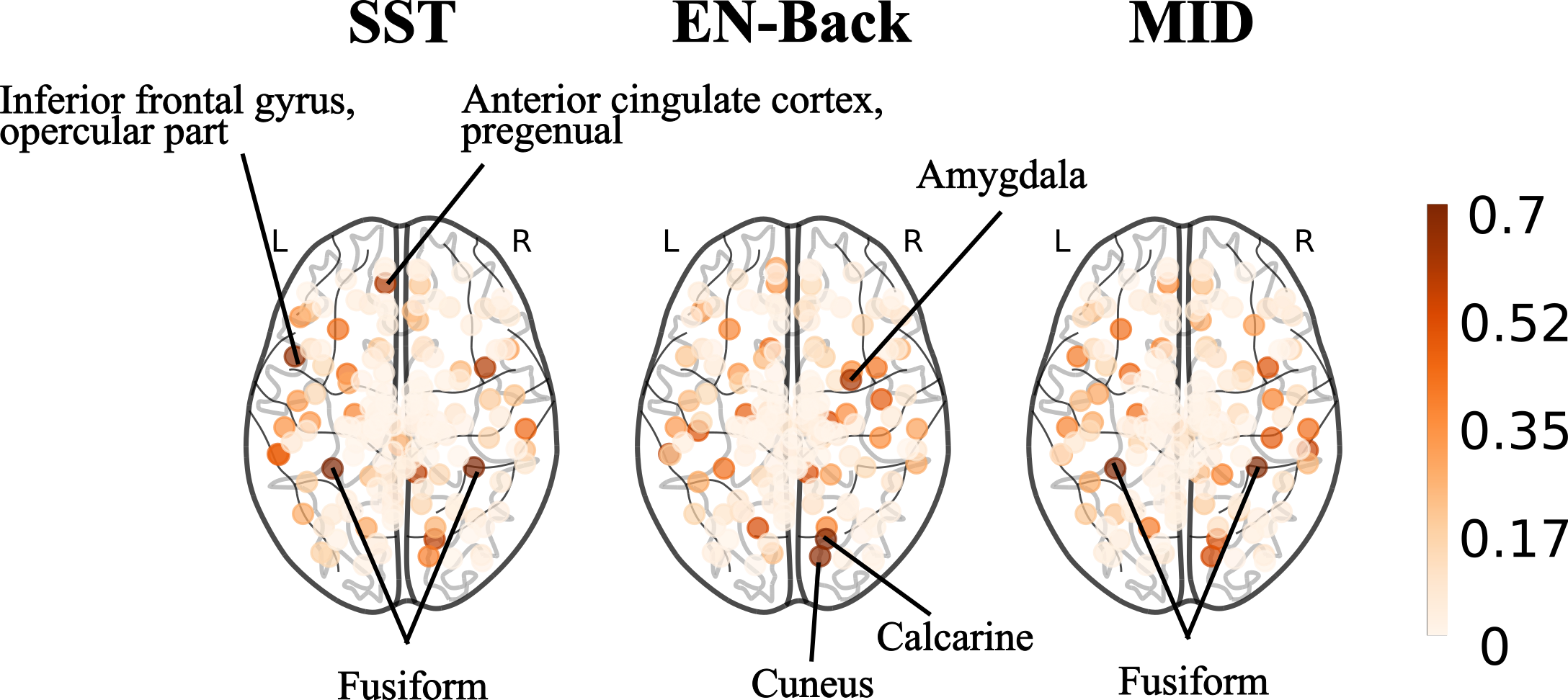}
    \vskip 0.1in
    \caption{Visualization of the frequency of each region under different fMRI tasks.}
    \label{fig:brainplot}
    \vskip -0.1in
\end{figure}
To better understand the roles of each brain region in cognition under different fMRI tasks, we study the frequency at which each region appears in a hyperedge out of all identified hyperedges. The frequency, which can be considered as a measure of region importance, is visualized in Figure \ref{fig:brainplot}.
Visual regions (\textit{Fusiform}, \textit{Cuneus}, \textit{Calcarine}) are especially active due to the intensive visual demands in all three tasks.
We found that the \textit{Inferior frontal gyrus, opercular part} and the \textit{Anterior cingulate cortex, pregenual}, recognized for their participation in response inhibition \citep{pornpattananangkul2016cultural}, frequently appear in the SST task. This aligns with what the SST task was designed to test. Interestingly, of the three tasks (SST, EN-back, MID), only EN-back prominently involves the \textit{Amygdala}, a key region for emotion processing. This makes sense as EN-back is the only task related to emotion.

\paragraph{Resting-State Brain Region Importance}
We visualize the region importance of the resting state in Figure \ref{fig:rest_region}. Different from task states, where specific brain regions are activated in response to particular tasks, brain activities during resting states, are not driven by external tasks, leading to more diffuse and less predictable patterns of activation. This makes it harder to pinpoint specific interactions or functions.
\begin{figure}[H]
    \centering
    \includegraphics[width=0.4\linewidth]{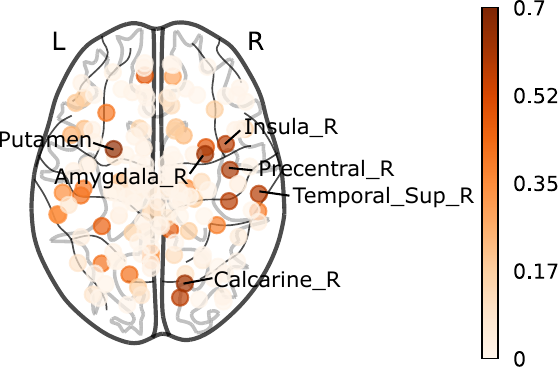}
    \caption{Region importance of resting state.}
    \label{fig:rest_region}
\end{figure}

% Sensory Integration: The presence of primary sensory areas (like the Calcarine, associated with vision) and association areas suggests that sensory information is being integrated with higher-order cognitive processes, perhaps as the subject views and responds to emotional images or words.
% \section{Limitations}
% \label{app:limitations}
% \mname{} only considers static high-order relations. Given that ABCD tasks are dynamic, including temporal changes and interactions, it will be interesting to study the evolution of these high-order relationships.

\section{Discussion of the Possible Methods to Interpret Hyperedges}
\label{app:interpret}
As mentioned in the main text, high-order relationships are much harder to interpret than pairwise ones given the exponential complexity. Here we propose a potential hierarchical strategy that tries to interpret them.

\begin{itemize}
    \item  \textbf{Adaptive Clustering Algorithm with Adjustable Granularity}: Suppose we have a clustering algorithm, where the granularity (or number of clusters) can be controlled. The clustering algorithm can cluster brain regions based on their functions (e.g. motor, visual, …) to different function modules \cite{shen2010graph}.
    \item   \textbf{Initial Analysis at Lower Granularity}: Set the granularity at a low level (i.e. small number of clusters). In this case, the high-order relationships are much easier to interpret (since the degree is low). However, some high-order relationships will connect to the same set of clusters, thus implying the interactions between the same set of function modules. For example, high-order relations $h_1$ and $h_2$ both connect to visual, motor, and emotion.
    \item \textbf{Refinement through Increased Granularity}: Find the high-order relationships that connect to the same set of function modules, and increase the granularity level, so we can tell the difference between them. For example, $h1$ may connect to the left part of the visual module, while $h_2$ connects to the right part.
\end{itemize}